\documentclass[11pt]{article}
\usepackage[utf8]{inputenc}
\usepackage[english]{babel}
\usepackage{epsfig,enumerate,amsmath,amsfonts,amsbsy,amssymb,amsthm,mathrsfs,lineno,ifpdf}
\usepackage{indentfirst,relsize,capt-of}
\usepackage[mathscr]{euscript}
\usepackage[numbers]{natbib}
\usepackage{setspace,graphicx}
\usepackage{latexsym}
\usepackage[usenames,dvipsnames]{pstricks}

 \usepackage{pst-grad} 
 \usepackage{pst-plot} 
\ifpdf
\usepackage[%
  pdftitle={Instructions for use of the document class
    elsart},%
  pdfauthor={Simon Pepping},%
  pdfsubject={The preprint document class elsart},%
  pdfkeywords={instructions for use, elsart, document class},%
  pdfstartview=FitH,%
  bookmarks=true,%
  bookmarksopen=true,%
  breaklinks=true,%
  colorlinks=true,%
  linkcolor=blue,anchorcolor=blue,%
  citecolor=blue,filecolor=blue,%
  menucolor=blue,pagecolor=blue,
  refcolor=blue,pagecolor=blue%
  urlcolor=blue]{hyperref}
\else
\usepackage[%
  breaklinks=true,%
  colorlinks=true,%
  linkcolor=blue,anchorcolor=blue,%
  citecolor=blue,filecolor=blue,%
  menucolor=blue,pagecolor=blue,%
  urlcolor=blue]{hyperref}
\fi
\usepackage{tikz,pgf} 
\usetikzlibrary{positioning,shapes,shadows,arrows} 
\usepackage[lined,ruled,linesnumbered]{algorithm2e}
\usepackage[margin=1.2in]{geometry}
\newtheorem{theorem}{Theorem}
\newtheorem{lemma}{Lemma}

\newtheorem{fact}{Fact}
\newtheorem{coro}{Corollary}
\newtheorem{obs}{Observation}

\newtheorem{claim}{Claim}
\newtheorem{subclaim}{Claim}[claim]

\let\oldenumerate\enumerate
\renewcommand{\enumerate}{
  \oldenumerate
  \setlength{\itemsep}{1.5pt}
  \setlength{\parskip}{0pt}
  \setlength{\parsep}{0pt}
}

\newcommand{\ntd}{\textsc{NTD}}

\newcommand{\bitem}{\item[$\bullet$]}

\graphicspath{{images/}{../images/}}

\usepackage{subfiles}

\usepackage{blindtext}

\begin{document}
\title{Algorithm and hardness results on neighborhood total domination in graphs\footnote{An extended abstract of this paper was presented in CTW 2018.}}

\author{Anupriya Jha, D. Pradhan\thanks{Corresponding author}, S. Banerjee  \\ \\
Department of Applied Mathematics \\
Indian Institute of Technology (ISM), Dhanbad\\
\small \tt Email: jha.anupriya@gmail.com; dina@iitism.ac.in; sumanta.banerjee5@gmail.com\\
\\
}

\date{}
\maketitle
%
%
%
%
%


\begin{abstract}
A set $D\subseteq V$ of a graph $G=(V,E)$ is called a neighborhood total dominating set of $G$ if $D$ is a dominating set and the subgraph of $G$ induced by the open neighborhood of $D$ has no isolated vertex. Given a graph $G$, \textsc{Min-NTDS} is the problem of finding a neighborhood total dominating set of $G$ of minimum cardinality. The decision version of \textsc{Min-NTDS} is known to be \textsf{NP}-complete for bipartite graphs and chordal graphs. In this paper, we extend this \textsf{NP}-completeness result to undirected path graphs, chordal bipartite graphs, and planar graphs. We also present a  linear time algorithm for computing a minimum neighborhood total dominating set in proper interval graphs. We show that for a given graph $G=(V,E)$, \textsc{Min-NTDS} cannot be approximated within a factor of $(1-\varepsilon)\log |V|$, unless \textsf{NP$\subseteq$DTIME($|V|^{O(\log \log |V|)}$)} and can be approximated within a factor of $O(\log \Delta)$, where $\Delta$ is the maximum degree of the graph $G$. Finally, we show that \textsc{Min-NTDS} is \textsf{APX}-complete for graphs of degree at most $3$.
\end{abstract}
\noindent
{\small \textbf{Keywords:}Domination; total domination; neighborhood total domination; polynomial time algorithm; \textsf{NP}-complete; \textsf{APX}-complete.}


\section{Introduction}
 \label{sec_intro}
 A set $D$ of
vertices of a graph $G =(V,E)$ is a \emph{dominating} set of $G$ if every vertex in $V\setminus D$ is adjacent to some vertex in $D$. The \emph{domination number} of a graph $G$, denoted by $\gamma(G)$, is the minimum cardinality of a dominating set of $G$. Domination and its variations have
many applications and have been widely studied in literature (see \cite{haynesb2,haynesb1}). Among the variations of domination, total domination in graphs is one of those. A set $D$ of vertices
of a graph $G=(V,E)$ is a \emph{total dominating} set of $G$ if every vertex in $V$ is adjacent to at least one vertex of $D$.  The \emph{total domination number} of a graph $G$, denoted by $\gamma_{\rm t}(G)$, is the minimum cardinality of a total dominating set of $G$. For extensive literature and
survey of total domination in graphs, we refer to \cite{henning-tds1,henning-book}.

In a graph $G=(V,E)$, the sets $N_G(v)=\{u\in V: uv\in E\}$ and $N_G[v]=N_G(v)\cup\{v\}$ denote the \emph{open neighborhood} and the \emph{closed neighborhood} of a vertex $v$, respectively. For a set $S\subseteq V$, $N_G(S)=\cup_{u\in S}N_G(u)$ and $N_G[S]=N_G(S)\cup S$. A total dominating set $D$ of a graph $G$ can be seen as a dominating set $D$ of $G$ such that the induced subgraph $G[D]$ has no isolated vertex. Looking the similar property of the open neighborhood of a dominating set $D$, Arumugam and Sivagnanam \cite{arumugam-ntd1} introduced the concept of neighborhood total domination in graphs. Formally, a dominating set $D$ of a graph $G$ is called a \emph{neighborhood total dominating} set, abbreviated a \ntd-set if $G[N_G(D)]$, i.e., the subgraph of $G$ induced by $N_G(D)$ has no isolated vertex. The \emph{neighborhood total domination number}, denoted by $\gamma_{\rm nt}(G)$, is the minimum cardinality of a \ntd-set of $G$.

Notice that in any graph without isolated vertices, every \ntd-set is a dominating set and every total dominating set is a \ntd-set. So the following observation follows.

\begin{obs}[\cite{arumugam-ntd1}]\label{obs1}
For any graph $G$ without any isolated vertex, $\gamma(G)\leq \gamma_{\rm nt}(G)\leq \gamma_{\rm t}(G)$.
\end{obs}

Observation \ref{obs1} motivates researchers to study on the neighborhood total domination in graphs since the neighborhood total domination number lies between the domination number and the total domination number, the two arguably most important domination parameters in graphs. Henning and Rad \cite{henning-ntd1} continued the further study on neighborhood total domination in graphs and gave several bounds on the neighborhood total domination number. Henning and Wash \cite{henning-ntd2} characterized the trees with large neighborhood total domination number. Mojdeh et al. \cite{mojdeh-ntd} studied the neighborhood total domination related to a graph and its complement. Recently, the algorithmic complexity of \textsc{Min-NTDS} has been studied by Lu et al. \cite{lu-ntd}. In particular, Lu et al. \cite{lu-ntd} proved that the decision version of \textsc{Min-NTDS} is \textsf{NP}-complete for bipartite graphs and chordal graphs and presented a linear time algorithm for computing a minimum \ntd-set in trees.  In this paper, we continue the study on the algorithmic complexity of \textsc{Min-NTDS} in graphs. In particular, the results that have been presented in this paper are summarized as follows.
\begin{enumerate}
\item In Section \ref{npc}, we first extend the known \textsf{NP}-completeness result of the decision version of \textsc{Min-NTDS} to undirected path graphs, chordal bipartite graphs, and planar graphs.

\item In Section \ref{sec_pig}, we present a linear time algorithm for computing a minimum \ntd-set in proper interval graphs, a subclass of chordal graphs.
\item  We in Section \ref{sec_hard} show that, unless \textsf{NP $\subseteq$ DTIME($n^{O(\log\log n)}$)}, \textsc{Min-NTDS} cannot be approximated
within a factor of $(1-\varepsilon) \ln n$ for any $\varepsilon>0$. Then we present an $O(\log \Delta)$-factor approximation algorithm for \textsc{Min-NTDS} in graphs of degree at most $\Delta$.

\item We in Section \ref{sec_apx} show that \textsc{Min-NTDS} is \textsf{APX}-complete for graphs with degree at most $3$.
\end{enumerate}

\section{Basic terminologies and preliminary results}
Given a graph $G=(V,E)$, we let $V=V(G)$ and $E=E(G)$ as the vertex set and edge set of $G$, respectively.  The \emph{degree} of a vertex $v$ in a graph $G$ is $|N_G(v)|$ and is denoted by $d_G(v)$. If $d_G(v)=1$, then $v$ is called a \emph{pendant} vertex of $G$. A vertex $v$ of $G$ is called a \emph{support} vertex of $u$ if $v\in N_G(u)$ and $u$ is a pendant vertex of $G$. A vertex $v$ is said to \emph{dominate} a vertex $u$ if $u\in N_G[v]$. For $S \subseteq V$, let $G[S]$ denote the subgraph of $G$ induced by $S$. For $S \subseteq V(G)$, if $G[S]$ is a complete subgraph of $G$, then $S$ is called a \emph{clique} of $G$ and if $G[S]$ has no edge, then $S$ is called an \emph{independent set} of $G$.

A graph $G$ is called a \emph{chordal} graph if every cycle of $G$ of length at least 4 has a \emph{chord}, i.e., an edge joining two non-consecutive vertices of the cycle. A graph $G=(V, E)$ is said to be a \emph{bipartite graph} if $V(G)$ can be partitioned into two disjoint independent sets. A bipartite graph $G$ is called a \emph{chordal bipartite graph} if every cycle in $G$ of length at least 6 has a chord.

A graph $G=(V,E)$ is called an \emph{intersection graph} for a finite family $\mathcal{F}$ of subsets of a nonempty set if there is a one-to-one correspondence between $\mathcal{F}$ and $V$ such that two sets in $\mathcal{F}$ have nonempty intersection if and only if their corresponding vertices in $V$ are adjacent. If $\mathcal{F}$ is a family of intervals on a real line such that no interval in $\mathcal{F}$ properly contains another interval in $\mathcal{F}$, then the intersection graph for $\mathcal{F}$ is called a \emph{proper interval} graph for $\mathcal{F}$. If $\mathcal{F}$ is a family of paths of a tree, then the intersection graph for $\mathcal{F}$ is called an \emph{undirected path} graph for $\mathcal{F}$.

Given an ordering $\sigma=(v_1,v_2,\ldots,v_n)$ of vertices of a graph $G$, for every $1\leq i\leq n$, let $G_i$ denote the graph $G[\{v_i,v_{i+1},\ldots,v_n\}]$. A vertex $v$ of a graph $G$ is called a \emph{simplicial} vertex of $G$ if $N_G[v]$ is a clique of $G$. An ordering $\sigma=(v_1,v_2,\ldots,v_n)$ is a \emph{perfect elimination ordering} (PEO) of $G$ if $v_i$ is a simplicial vertex of $G_i=G[\{v_i,v_{i+1},\ldots,v_n\}]$ for all $i, 1\leq i\leq n$. A PEO $\sigma=(v_1,v_2,\ldots,v_n)$ of a chordal graph is a \emph{bi-compatible elimination ordering}(BCO) if $\sigma^{-1}=(v_n,v_{n-1},\ldots,v_1)$, i.e. the reverse of $\sigma$, is also a PEO of $G$. This implies that $v_i$ is simplicial in $G[\{v_1,v_2,\ldots,v_i\}]$ as well as in $G[\{v_i,v_{i+1},\ldots,v_n\}]$. The proper interval graphs are characterized in terms of BCO \cite{laskar-pig}. In \cite{laskar-pig}, it has also been shown that if $\sigma=(v_1,v_2,\ldots,v_n)$ is a BCO of a connected graph $G$, then $P=v_1v_2\cdots v_n$ is a Hamiltonian path (a path containing all the vertices) of $G$. The following is an observation about a BCO of a proper interval graph.

\begin{obs}[\cite{panda-bco}]\label{obsprelim2}
Suppose $\sigma=(v_1,v_2,\ldots,v_n)$ is a BCO of a connected proper interval graph $G=(V,E)$. Then the following are true.
\begin{enumerate}
\item If $v_iv_j\in E$ for some $i<j$, then $\{v_i,v_{i+1},\ldots,v_j\}$ is a clique of $G$.

\item If $i<k<j$ and $v_iv_j\in E$, then $N_G[v_k]\cap \{v_1,v_2,\ldots,v_k\}\subseteq N_G[v_i]$ and $N_G[v_k]\cap \{v_{k},v_{k+1},\ldots,v_n\}\subseteq N_G[v_j]$.
\end{enumerate}
\end{obs}

\begin{obs}\label{obsprelim5}
Suppose $\sigma=(v_1,v_2,\ldots,v_n)$ is a BCO of a connected proper interval graph $G$. Then the following are true.

\begin{enumerate}
\item If $k\leq i$, then $N_G[v_k]\cap V(G_i)\subseteq N_{G_i}[v_i]$.

\item If $v_k\in N_{G_i}[v_i]$, then $N_{G_i}[v_i]\subseteq N_{G_i}[v_k]$.
\end{enumerate}

\end{obs}
\begin{proof}
(a) If $k=i$, then we are done. So assume that $k<i$ and let $v_r\in N_G[v_k]\cap V(G_i)$. By Observation \ref{obsprelim2}(a), $\{v_k,v_{k+1},\ldots,v_r\}$ is a clique. So $v_r\in N_G[v_i]$ and hence $N_{G_i}[v_i]$, completing the proof of (a).

(b) If $k=i$, then we are done. So assume that $k\neq i$ and let $v_r\in N_{G_i}[v_i]$. Since $\sigma$ is a BCO of $G$, $v_i$ is a simplicial vertex in $G_i$. This implies that $v_r\in N_G[v_k]$ and hence $v_r\in N_{G_i}[v_k]$, consequently, $N_{G_i}[v_i]\subseteq N_{G_i}[v_k]$. This completes the proof of (b).
\end{proof}
\begin{obs}\label{obsprelim4}
Let $G$ be a connected proper interval graph having $n$ vertices and $v$ be a vertex of $G$. Then either $N_G[v]$ induces a path of length at most $2$ or $v$ is contained in a clique of $G$ of size at least $3$.
\end{obs}
\begin{proof}
Let $\sigma=(v_1,v_2,\ldots,v_n)$ be a BCO of $G$. If $v=v_1$ or $v=v_n$, then either $N_G[v]$ induces a path of length $1$ (if $d_G(v)=1$) or $v$ is contained in  a clique of size at least $3$ (if $d_G(v)\geq 2$) and hence we are done. Hence we may assume that $v\neq v_1$ and $v\neq v_n$, and so $v=v_i$ for some $i\neq 1$ and $i\neq n$. Since $\sigma$ is a BCO of $G$, $v_i$ is simplicial in $G[\{v_1,v_2,\ldots,v_i\}]$ as well as $G[\{v_i,v_{i+1},\ldots,v_n\}]$. If $|N_G(v_i)\cap \{v_1,v_2,\ldots,v_i\}|\geq 2$ or $|N_G(v_i)\cap \{v_i,v_{i+1},\ldots,v_n\}|\geq 2$, then $v_i$ is contained in a clique of size at least $3$. So assume that $N_G(v_i)\cap \{v_1,v_2,\ldots,v_i\}=\{v_{i-1}\}$ and $N_G(v_i)\cap \{v_i,v_{i+1},\ldots,v_n\}=\{v_{i+1}\}$.   If $v_{i-1}v_{i+1}\in E(G)$, then $v_i$ is contained in a clique of size at least $3$. If $v_{i-1}v_{i+1}\notin E(G)$, then $N_G[v_i]$ induces a path of length $2$ in $G$.
\end{proof}

We conclude this section by presenting the following observations which are easy to follow and hence the proofs are omitted.

\begin{obs}\label{obsprelim1}
Let $G=(V,E)$ be a connected graph with at most two vertices, then $\gamma_{\rm nt}(G)=|V|$.
\end{obs}

\begin{obs}\label{obsprelim3}
Suppose $D$ is a dominating set of a graph $G$ and $v$ is a vertex of $G$ such that $v\in D$. Then $v$ is not an isolated vertex in $G[N_G(D)]$.
\end{obs}
\begin{obs}\label{obsprelim0}
Suppose $D$ is a \ntd-set of a connected graph $G$. If $v$ is a pendant vertex of $G$ and $u$ is the support vertex of $v$, then either $v\in D$ or $|N_G[u]\cap D|\geq 2$.
\end{obs}

\section{\textsf{NP}-completeness}\label{npc}

Given a graph $G$, we denote \textsc{Min-Dom-Set} as the problem of finding a minimum dominating set of $G$. In this section, we provide a polynomial time reduction from the decision version of \textsc{Min-Dom-Set} to the decision version of \textsc{Min-NTDS} and using this reduction, we prove that the decision version of \textsc{Min-NTDS} is \textsf{NP}-complete for undirected path graphs, chordal bipartite graphs, and planar graphs.

\begin{figure}[h]
\begin{center}
\includegraphics[scale=.8]{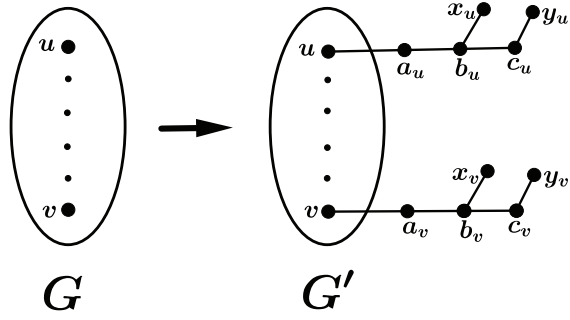}
\caption{\label{npc-fig} The construction of the graph $G'$ from a graph $G$}
\end{center}
\end{figure}

Let $G$ be a graph having $n$ vertices. We construct a new graph $G'$, where $V(G')=V(G)\cup\{a_v,b_v,c_v,x_v,y_v: v\in V(G)\}$ and $E(G')=E(G)\cup\{va_v,a_vb_v,b_vc_v,b_vx_v,c_vy_v\}$. Notice that for every $v\in V(G)$, $x_v$ and $y_v$ are pendant vertices of $G'$. The construction of the graph $G'$ is shown in Figure \ref{npc-fig}.

\begin{lemma}\label{npc-lem1}
The graph $G$ has a dominating set of cardinality at most $k$ if and only if the graph $G'$ has a \ntd-set of cardinality at most $k+2n$.
\end{lemma}
\begin{proof}
Suppose $D$ is a dominating set of $G$ of cardinality at most $k$. Let $S=D\cup\{b_v,c_v: v\in V(G)\}$. It is clear that $S$ is a dominating set of $G'$. Since $b_v,c_v\in S$, by Observation \ref{obsprelim3}, $b_v$ and $c_v$ are not isolated vertices in $G'[N_{G'}(S)]$. Again since $b_vc_v\in E(G')$ and $b_v,c_v\in S$, the vertices $a_v,x_v,y_v$ are not isolated vertices in $G'[N_{G'}(S)]$. Let $v\in V(G)$. If $v\in S$, then by Observation \ref{obsprelim3}, $v$ is not an isolated vertex in $G'[N_{G'}(S)]$. If $v\notin S$, then $v$ is not an isolated vertex in $G'[N_{G'}(S)]$ since $a_v\notin S$. So $S$ is a \ntd-set of $G'$ of cardinality at most $k+2n$.

Now assume that $S'$ is a \ntd-set of $G'$ of cardinality $k+2n$. We first prove the following claim.
\begin{claim}\label{npc-claim1}
There is a \ntd-set $S''$ of $G'$ with $|S''|\leq |S'|$ such that $b_v,c_v\in S''$ and $a_v,x_v,y_v\notin S''$ for every $v\in V(G)$.
\end{claim}

\begin{proof}[Proof of Claim \ref{npc-claim1}:]
If $b_v,c_v\in S'$ and $a_v,x_v,y_v\notin S'$ for every $v\in V(G)$, then we take $S''=S'$ and hence we are done. Let $v\in V(G)$ be arbitrary. We consider the following cases.

\noindent\textbf{Case 1:} $x_v\in S'$ and $y_v\notin S'$.

Since $y_v\notin S'$, by Observation \ref{obsprelim0}, $|N_{G'}[c_v]\cap S'|\geq 2$ which implies that $b_v,c_v\in S'$. Now $S''=(S'\setminus\{x_v\})$ is a \ntd-set of $G'$ such that if $y_v\notin S''$, then $b_v,c_v\in S''$ and $x_v\notin S''$. Notice that $|S''|\leq |S'|$.

\noindent\textbf{Case 2:} $x_v,y_v\in S'$.

If $b_v,c_v\in S'$, then $S''=S'\setminus\{x_v,y_v\}$ is a smaller \ntd-set of $G'$ than $S'$. Now assume that $b_v\notin S'$ or $c_v\notin S'$. If $b_v,c_v\notin S'$, then $S''=(S'\setminus\{x_v,y_v\})\cup\{b_v,c_v\}$ is a \ntd-set of $G'$ with $|S''|=|S'|$. Similarly if $c_v\notin S'$ and $b_v\in S'$ (resp. $b_v\notin S'$ and $c_v\in S'$), then $S''=(S'\setminus\{x_v,y_v\})\cup \{c_v\}$ (resp. $S''=(S'\setminus\{x_v,y_v\})\cup \{b_v\}$) is a \ntd-set of $G'$. Notice that in each case, $|S''|\leq |S'|$.

\noindent\textbf{Case 3:} $x_v\notin S'$ and $y_v\in S'$.

Since $x_v\notin S'$ and $x_v$ is a pendant vertex of $G'$, by Observation \ref{obsprelim0}, $|N_{G'}[b_v]\cap S'|\geq 2$. So either $a_v\in S'$ or $c_v\in S'$. Moreover, $b_v\in S'$ to dominate $x_v$. If $c_v\in S'$, then $S''=S'\setminus\{y_v\}$ is a \ntd-set of $G'$ with $|S''|<|S'|$. Now assume that $c_v\notin S'$. Then $a_v\in S'$. Now $S''=(S'\setminus\{y_v\})\cup\{c_v\}$ is a \ntd-set of $G'$ with $|S''|=|S'|$.

\noindent\textbf{Case 4:} $x_v,y_v\notin S'$.

Then $b_v,c_v\in S'$ to dominate $x_v$ and $y_v$. Let  $a_v\in S'$. If $v\in S'$, then $S''=(S'\setminus\{a_v\})$ is a \ntd-set of $G'$ with $|S''|<|S'|$. If $v\notin S'$, then $S''=(S'\setminus\{a_v\})\cup\{v\}$ is a \ntd-set of $G'$ with $|S''|=|S'|$. Notice that if $v\notin S'$, then $S'\setminus\{a_v\}$ may not be a \ntd-set of $G'$ if $N_{G}(v)\subset S'$.

So by Cases 1-4, we can get a \ntd-set $S''$ of $G'$ with $|S''|\leq |S'|$ such that $b_v,c_v\in S''$ and $a_v,x_v,y_v\notin S''$ for every $v\in V(G)$. This completes the proof of the claim. \end{proof}

By Claim \ref{npc-claim1}, we can get a \ntd-set $S''$ of $G'$ with $|S''|\leq |S'|$ such that $b_v,c_v\in S''$ and $a_v,x_v,y_v\notin S''$ for every $v\in V(G)$. Let $D'=S''\setminus\{b_v,c_v:v\in V(G)\}$. It can be seen that $D'$ is a dominating set of $G$ since $S''$ is a \ntd-set of $G'$. Notice that $|D'|=|S''|-2|V(G)|\leq |S'|-2n\leq k$. This completes the proof of the lemma.
\end{proof}

Notice that if $G$ is a chordal bipartite (resp. planar) graph, then $G'$ is also a chordal bipartite (resp. planar) graph. We now prove that if $G$ is an undirected path graph, then $G'$ is also an undirected path graph. 

Let $G$ be an undirected path graph. Then there is a tree $T$ and a set of paths $\{P_v
: v\in V(G)\}$ of $T$ such that $uv\in E(G)$ if and only if $P_u\cap P_v\neq\emptyset$. For each path $P_v$, consider an end point $v^*$ of $P_v$. Construct a tree
$T'$ which results from $T$ by introducing  two new paths,$v^*d_va_vb_vc_vy_v$ and $b_vz_vx_v$ joined at $v^*$ and $b_v$, respectively. Let $P'_v=P_v\cup\{v^*d_v\}$. We note that the path $P_v\cup\{v^*d_v\}$ corresponds to the vertex $v$ of $G$ in $T'$. Similarly, the paths $d_va_v,a_vb_vz_v,b_vc_v,z_vx_v,c_vy_v$ correspond to the vertices $a_v,b_v,c_v,x_v,y_v$, respectively in $T'$. Then it is
clear that $G'$ is the intersection graph of the set of paths $\displaystyle\bigcup_{v\in V(G)}\{P'_v\cup \{d_va_v,a_vb_vz_v,b_vc_v,z_vx_v,c_vy_v\}\}$. Hence $G'$ is an undirected path graph.

Since the decision version of \textsc{Min-Dom-Set} is \textsf{NP}-complete for undirected path graphs \cite{booth-ds}, for planar graphs \cite{garey}, and for chordal bipartite graphs \cite{muller}, by
Lemma 8, we have the following theorem.
\begin{theorem}
The decision version of \textsc{Min-NTDS} is \textsf{NP}-complete for undirected path graphs, chordal bipartite graphs, and planar graphs.
\end{theorem}

\section{Algorithm for computing a minimum \ntd-set in proper interval graphs} \label{sec_pig}

In this section, we present an algorithm to compute a minimum \ntd-set in a given proper interval graph $G$. The algorithm runs in $O(|V(G)|+|E(G)|)$ time. We first present some lemmas that will help in designing our algorithm for computing a minimum \ntd-set in a given connected proper interval graph. 

\begin{lemma}\label{obspig1}
Let $D$ be a dominating set of a connected proper interval graph $G$. Then for any $v\in V(G)$, if $v$ is contained in a clique of size at least $3$, then $v$ is not an isolated vertex in $G[N_G(D)]$.
\end{lemma}
\begin{proof}
Let $v$ be a vertex of $G$ that is contained in a clique of size at least $3$. This implies that there exist two vertices $x,y\in N_G(v)$ such that $xy\in E(G)$. If $v\in D$, then by Observation \ref{obsprelim3}, $v$ is not an isolated vertex in $G[N_G(D)]$ and hence we are done. So assume that $v\notin D$. If $x,y\in D$, then $v$ is not an isolated vertex in $G[N_G(D)]$. So assume that $x\notin D$ or $y\notin D$. Then either $x\in N_G(D)$ or $y\in N_G(D)$. So $v$ is not an isolated vertex in $G[N_G(D)]$.
\end{proof}

Let $G$ be a connected proper interval graph with a BCO $\sigma=(v_1,v_2,\ldots,v_n)$. For each $v_i, 1\leq i\leq n$, let $\ell(v_i)=\max\{\{i\}\cup\{k:v_iv_k\in E(G)$ and $k>i\}\}$.

\begin{obs}\label{obspig1-2}
Let $\sigma=(v_1,v_2,\ldots,v_n)$ be a BCO of a connected proper interval graph $G$. For every $i<j$, $\ell(v_i)\leq \ell(v_j)$.
\end{obs}

\begin{lemma}\label{pig2}
Let $G$ be a connected proper interval graph and $\sigma=(v_1,v_2,\ldots,v_n)$ be a BCO of $G$. If $d_G(v_1)\geq 2$ and $\ell(v_1)=j$, then there is a minimum \ntd-set of $G$ containing $v_j$.
\end{lemma}

\begin{proof}
Let $D$ be a minimum \ntd-set of $G$ and $k$ be the minimum index with respect to $\sigma$ such that $v_k\in D$ and $v_k$ dominates $v_1$. Since $\ell(v_1)=j$, $k\leq j$. If $k=j$, then we are done. So assume that $k<j$. Let $D'=(D\setminus\{v_k\})\cup \{v_j\}$. We now prove that $D'$ is also a \ntd-set of $G$. Since $\sigma$ is a BCO of $G$, $\{v_1,v_2,\ldots,v_j\}$ is a clique by Observation \ref{obsprelim2}(a) and hence $v_kv_j\in E(G)$. Since $v_k$ is a simplicial vertex in $G_k=G[\{v_k,v_{k+1},\ldots,v_n\}]$ and $v_j\in N_{G_k}(v_k)$, by Observation \ref{obsprelim5}(b), $N_{G_k}[v_k]\subseteq N_{G_k}[v_j]$. This implies that $N_{G}[v_k]\subseteq N_{G}[v_j]$ since $\{v_1,v_2,\ldots,v_j\}$ is a clique containing $v_k$ and $v_j$. So $D'$ is a dominating set of $G$. To prove that $G[N_G(D')]$ has no isolated vertex, it is sufficient to show that no vertex of $N_G(v_j)$ is an isolated vertex in $G[N_G(D')]$. Since $|N_G(v_j)\cap \{v_1,v_2,\ldots,v_{j-1}\}|\geq 2$, by Lemma  \ref{obspig1}, each vertex of $N_G(v_j)\cap \{v_1,v_2,\ldots,v_{j-1}\}$ is not an isolated vertex in $G[N_G(D')]$. So we show that no vertex of $N_{G_j}(v_j)$ is an isolated vertex in $G[N_G(D')]$ or $D'$ can be modified to a minimum \ntd-set $D''$ of $G$ such that $v_j\in D''$. We proceed further with the following series of claims.

\begin{claim}\label{pig2claim1.1}
If $|N_{G_j}(v_j)|\neq 1$, then no vertex of $N_{G_j}(v_j)$ is an isolated vertex in $G[N_G(D')]$.
\end{claim}
\begin{proof}[Proof of Claim \ref{pig2claim1.1}:]
If $N_{G_j}(v_j)=\emptyset$, then no vertex of $N_{G_j}(v_j)$ is an isolated vertex in $G[N_G(D')]$ in this case. If $|N_{G_j}(v_j)|\geq 2$, then by Lemma \ref{obspig1}, no vertex of $N_{G_j}(v_j)$ is an isolated vertex in $G[N_G(D')]$.
\end{proof}

If  $|N_{G_j}(v_j)|\neq 1$, then by Claim \ref{pig2claim1.1}, no vertex of $N_{G_j}(v_j)$ is an isolated vertex in $G[N_G(D')]$ implying that $D'$ is a minimum \ntd-set of $G$ containing $v_j$. So assume that $|N_{G_j}(v_j)|=1$. Since $\sigma$ is a BCO of $G$, $N_{G_j}(v_j)=\{v_{j+1}\}$. Let $p$ be the minimum index with respect to $\sigma$ such that $p>j$, $G[\{v_j,v_{j+1},\ldots,v_p\}]$ is a path, and $|N_{G_p}(v_p)|\neq 1$.

\begin{claim}\label{pig2claim1.2}
If $p=j+1$, then no vertex of $N_{G_j}(v_j)$ is an isolated vertex in $G[N_G(D')]$.
\end{claim}
\begin{proof}[Proof of Claim \ref{pig2claim1.2}:]
If $p=j+1$, then either $j+1=n$ or $|N_{G_{j+1}}(v_{j+1})|\geq 2$.
If $|N_{G_{j+1}}(v_{j+1})|\geq 2$, then by Lemma \ref{obspig1}, $v_{j+1}$ is not an isolated vertex in $G[N_G(D')]$. If $N_{G_{j+1}}(v_{j+1})=\emptyset$, then $j+1=n$. If $d_G(v_n)\geq 2$, then by Lemma \ref{obspig1}, $v_n=v_{j+1}$ is not an isolated vertex in $G[N_G(D')]$. So we may assume that $N_G(v_n)=\{v_j\}$. Since $v_j\notin D$,  $v_{n}$ must be in $D$. Moreover, $v_n\in D'$ since $D'=(D\setminus\{v_k\})\cup \{v_j\}$. So by Observation \ref{obsprelim3}, $v_{j+1}=v_n$ is not an isolated vertex in $G[N_G(D')]$, completing the proof of the claim.
\end{proof}

\begin{claim}\label{pig2claim1.3}
If $p=j+2$, then either no vertex of $N_{G_j}(v_j)$ is an isolated vertex in $G[N_G(D')]$ or there is a minimum \ntd-set of $G$ containing $v_j$.
\end{claim}
\begin{proof}[Proof of Claim \ref{pig2claim1.3}:]
If $p=j+2$, then either $j+2=n$ or $|N_{G_{j+2}}(v_{j+2})|\geq 2$. Recall that $N_{G_j}(v_j)=\{v_{j+1}\}$.  If $v_{j+1}\in D'$, then no vertex of $N_{G_j}(v_j)$ is an isolated vertex in $G[N_G(D')]$ and hence we are done. So assume that $v_{j+1}\notin D'$. This implies that $v_{j+1}\notin D$ since $D'=(D\setminus\{v_k\})\cup\{v_j\}$. Since $v_j\notin D$ and $D$ is a dominating set of $G$, $v_{j+2}$ must be in $D$ to dominate $v_{j+1}$. Thus $v_{j+2}\in D'$.

If $j+2=n$, then $(D\setminus\{v_k,v_{j+2}\})\cup \{v_j,v_{j+1}\}$ is a minimum \ntd-set of $G$ containing $v_j$. If $|N_{G_{j+2}}(v_{j+2})|\geq 2$, then let $D''=(D\setminus\{v_k,v_{j+2}\})\cup \{v_j,v_{j+3}\}$. Notice that $D''$ is a dominating set of $G$. If $N_{G_{j+2}}(v_{j+2})\cap D\neq \emptyset$, then $(D\setminus\{v_k,v_{j+2}\})\cup \{v_j\}$ is a smaller \ntd-set of $G$ than $D$. This is a contradiction. So $N_{G_{j+2}}(v_{j+2})\cap D=\emptyset$. This implies that $N_{G_{j+2}}(v_{j+2})\cap D'=\emptyset$. Since $\sigma$ is a BCO of $G$, by Observation \ref{obsprelim2}(b), $N_G[v_{j+3}]\cap \{v_1,v_2,\ldots,v_{j+3}\}\subseteq N_G[v_{j+2}]$ and $N_{G}[v_{j+3}]\cap \{v_{j+3},v_{j+4},\ldots,v_n\}\subseteq N_G[v_{j+4}]$. Since $v_{j+2},v_{j+4}\notin D''$, no vertex of $N_G(v_{j+3})$ is an isolated vertex in $G[N_G(D'')]$. Since $D''=(D'\setminus\{v_{j+2}\})\cup\{v_{j+3}\}$, where $D'=(D\setminus\{v_k\})\cup \{v_j\}$, no vertex of $N_{G_j}(v_j)$ is an isolated vertex in $G[N_G(D'')]$. So $D''$ is a minimum \ntd-set of $G$ containing $v_j$.
\end{proof}

\begin{claim}\label{pig2claim1.4}
If $p=j+3$, then either no vertex of $N_{G_j}(v_j)$ is an isolated vertex in $G[N_G(D')]$ or there is a minimum \ntd-set of $G$ containing $v_j$.
\end{claim}
\begin{proof}[Proof of Claim \ref{pig2claim1.4}:]
By the choice of $p$, $G[\{v_j,v_{j+1},v_{j+2},v_{j+3}\}]$ is a path. Moreover, notice that $v_{j+1}\notin D$ and $v_{j+2}\in D$. Since $D'=(D\setminus\{v_k\})\cup\{v_j\}$, we have $v_{j+1}\notin D'$ and $v_{j+2}\in D'$.

If $j+3=n$, then $v_{j+3}\in D'$ which implies that $v_{j+2},v_{j+3}\in N_G(D')$. So $v_{j+1}$ is not an isolated vertex in $G[N_G(D')]$. If $N_{G_{j+3}}(v_{j+3})\geq 2$, then $v_{j+4},v_{j+5}\in N_{G_{j+3}}(v_{j+3})$. If $v_{j+3}\in D'$, then $D'$ is a minimum \ntd-set of $G$ containing $v_j$. So assume that  $v_{j+3}\notin D'$. If $N_{G_{j+3}}(v_{j+3})\cap D'\neq \emptyset$, then $(D'\setminus\{v_{j+2}\})\cup \{v_{j+1}\}$ is a minimum \ntd-set of $G$ containing $v_j$. If $N_{G_{j+3}}(v_{j+3})\cap D'= \emptyset$, then $(D'\setminus\{v_{j+2}\})\cup \{v_{j+3}\}$ is a minimum \ntd-set of $G$ containing $v_j$.
\end{proof}

\begin{claim}\label{pig2claim1.5}
If $G[\{v_j,v_{j+1},\ldots,v_p\}]$ is a path of length $3r$ for some $r\geq 1$, then either no vertex of $N_{G_j}(v_j)$ is an isolated vertex in $G[N_G(D')]$ or there is a minimum \ntd-set of $G$ containing $v_j$..
\end{claim}
\begin{proof}[Proof of Claim \ref{pig2claim1.5}:]
In this case, $p=j+3r$. If $r=1$, then $G[\{v_j,v_{j+1},v_{j+2},v_{j+3}\}]$ is a path and by Claim \ref{pig2claim1.4}, either no vertex of $N_{G_j}(v_j)$ is an isolated vertex in $G[N_G(D')]$ or there is a minimum \ntd-set of $G$ containing $v_j$. So assume that $r\geq 2$. Then without loss of generality, we can assume (similar to the way we have proceeded in Claim \ref{pig2claim1.4}) that $v_j,v_{j+2},v_{j+5},\ldots,v_{j+3r-1}\in D'$. If $p=j+3r=n$, then $v_{j+3r}\in D'$ or $v_{j+3r-2}\in D'$ since $v_{j+3r}$ is a pendant vertex of $G$. Without loss of generality, assume that $v_{j+3r}\in D'$. Let $S_1=\{v_{j+2},v_{j+5},\ldots,v_{j+3r-1}\}$, $S_2=\{v_{j+3},v_{j+6},\ldots,v_{j+3r-3}\}$, and $S_3=D'\cap S_2$. Notice that $|S_2|=|S_1|-1$. Now let $D''=(D'\setminus S_1)\cup (S_2\setminus S_3)$. Then $D''$ is a smaller \ntd-set of $G$ which is a contradiction since $|D''|=|D'|-|S_1|+|S_2|-|S_3|=|D'|-(|S_1|-|S_2|+|S_3|)<|D'|=|D|$. So $j+3r\neq n$. Then $|N_{G_{j+3r}}(v_{j+3r})|\geq 2$. Moreover, $v_{j+3r}\notin D'$; otherwise $(D'\setminus\{v_{j+2},v_{j+5},\ldots,v_{j+3r-1}\})\cup \{v_{j+3},v_{j+6},\ldots,v_{j+3r-3}\}$ is a smaller \ntd-set of $G$. If $N_{G_{j+3r}}(v_{j+3r})\cap D'\neq \emptyset$, then let $S_4=\{v_{j+1},v_{j+4},\ldots,v_{j+3r-2}\}$ and $S_5=D'\cap S_4$. Notice that $|S_4|=|S_1|$.  Now if $S_5\neq \emptyset$, then $D_1=(D'\setminus S_1)\cup (S_4\setminus S_5)$ is a smaller \ntd-set of $G$ than $D$ since $|D_1|=|D'|-|S_1|+|S_4|-|S_5|=|D'|-|S_5|<|D'|=|D|$. This is a contradiction. So $S_5= \emptyset$. Now $D'''=(D'\setminus S_1)\cup S_4$ is a  minimum \ntd-set of $G$ containing $v_j$ since $v_{j+1}\in D'''$. If $N_{G_{j+3r}}(v_{j+3r})\cap D'=\emptyset$, then let $S_6=\{v_{j+3},v_{j+6},\ldots,v_{j+3r}\}$ and $S_7=D'\cap S_6$. Notice that $|S_6|=|S_1|$. If $S_7\neq \emptyset$, then $(D'\setminus S_1)\cup (S_6\setminus S_7)$ is a smaller \ntd-set of $G$ than $D$ which is a contradiction. So $S_7=\emptyset$. Now $(D'\setminus S_1)\cup S_6$ is a minimum \ntd-set of $G$ containing $v_j$. This completes the proof of the claim. 
\end{proof}
\begin{claim}\label{pig2claim1.6}
If $G[\{v_j,v_{j+1},\ldots,v_p\}]$ is a path of length $3r+1$ for some $r\geq 0$, then either no vertex of $N_{G_j}(v_j)$ is an isolated vertex in $G[N_G(D')]$ or there is a minimum \ntd-set of $G$ containing $v_j$.
\end{claim}
\begin{proof}[Proof of Claim \ref{pig2claim1.6}:]
If $r=0$, then $p=j+1$. So by Claim \ref{pig2claim1.2}, no vertex of $N_{G_j}(v_j)$ is an isolated vertex in $G[N_G(D')]$ and hence we are done. So assume that $r\geq 1$. Clearly $p=j+3r+1$.  Without loss of generality, we can assume that $v_j,v_{j+2},v_{j+5},\ldots,v_{j+3r-1}\in D'$. Let $S_1=\{v_{j+2},v_{j+5},\ldots,v_{j+3r-1}\}$, $S_2=\{v_{j+3},v_{j+6},\ldots,v_{j+3r-3}\}$, and $S_3=D'\cap S_2$.

If $p=j+3r+1=n$, then, since $v_{j+3r+1}$  is a pendant vertex of $G$, $v_{j+3r}$ must be in $D$ as $v_{j+3r-1}\in D$ (since $D'=(D\setminus\{v_k\})\cup \{v_j\}$). 
Notice that $|S_2|=|S_1|-1$. If $S_3\neq \emptyset$, then let $D_1=(D'\setminus S_1)\cup ((S_2\setminus S_3)\cup\{v_{j+3r+1}\})$. Notice that $D_1$ is a \ntd-set of $G$ with $|D_1|=|D'|-|S_1|+|S_2|+1-|S_3|=|D'|-|S_3|<|D'|=|D|$. This is a contradiction. So $S_3=\emptyset$.  Then $(D'\setminus S_1)\cup(S_2\cup\{v_{j+3r+1}\})$ is a minimum \ntd-set of $G$ containing $v_j$. 

If $j+3r+1\neq n$, then $|N_{G_{j+3r+1}}(v_{j+3r+1})|\geq 2$. To dominate $v_{j+3r+1}$, either $v_{j+3r}\in D$ or $N_{G_{j+3r+1}}[v_{j+3r+1}]\cap D\neq \emptyset$. If $v_{j+3r}\in D$ and $N_{G_{j+3r+1}}[v_{j+3r+1}]\cap D\neq \emptyset$, then let $D_2=(D'\setminus S_1)\cup(S_2\setminus S_3)$. Notice $D_2$ is a \ntd-set of $G$ with $|D_2|<|D|$. This is a contradiction. If $N_{G_{j+3r+1}}[v_{j+3r+1}]\cap D\neq \emptyset$ and $v_{j+3r}\notin D$, then let $D_3=(D'\setminus S_1)\cup(S_2\setminus S_3)\cup\{v_{j+3r}\}$. If $S_3\neq \emptyset$, then $D_3$ is a \ntd-set of $G$ with $|D_3|<|D|$. This is again a contradiction to the minimality of $D$. So $S_3=\emptyset$ and hence $D_3$ is a minimum \ntd-set of $G$ containing $v_j$. If $N_{G_{j+3r+1}}[v_{j+3r+1}]\cap D=\emptyset$, then $v_{j+3r}\in D$. Now let $D_4=(D'\setminus S_1)\cup (S_2\setminus S_3)\cup \{v_{j+3r+1}\}$. If $S_3\neq \emptyset$, then $D_4$ is a \ntd-set of $G$ with $|D_4|<|D|$. This is again a contradiction to the minimality of $D$. So $S_3=\emptyset$ and hence $D_4$ is a minimum \ntd-set of $G$ containing $v_j$. This completes the proof of the claim. 
\end{proof}

\begin{claim}\label{pig2claim1.7}
If $G[\{v_j,v_{j+1},\ldots,v_p\}]$ is a path of length $3r+2$ for some $r\geq 0$, then either no vertex of $N_{G_j}(v_j)$ is an isolated vertex in $G[N_G(D')]$ or there is a minimum \ntd-set of $G$ containing $v_j$.
\end{claim}
\begin{proof}[Proof of Claim \ref{pig2claim1.7}:]
If $r=0$, then $p=j+2$. So by Claim \ref{pig2claim1.3}, either no vertex of $N_{G_j}(v_j)$ is an isolated vertex in $G[N_G(D')]$ or there is a minimum \ntd-set of $G$ containing $v_j$ and hence we are done. So assume that $r\geq 1$. Clearly $p=j+3r+2$. Without loss of generality, we can assume that $v_j,v_{j+2},v_{j+5},\ldots,v_{j+3r-1},v_{j+3r+2}\in D'$. Let $S_1=\{v_{j+2},v_{j+5},\ldots,v_{j+3r-1},v_{j+3r+2}\}$, $S_2=\{v_{j+3},v_{j+6},\ldots,v_{j+3r},v_{j+3r+1}\}$, and $S_3=D'\cap S_2$. Notice that $|S_1|=|S_2|$.

If $p=j+3r+2=n$, then let $D_1=(D'\setminus S_1)\cup (S_2\setminus S_3)$. If $S_3\neq \emptyset$, then $D_1$ is a \ntd-set of $G$ with $|D_1|<|D|$. This is again a contradiction to the minimality of $D$. So $S_3=\emptyset$ and hence $D_1$ is a minimum \ntd-set of $G$ containing $v_j$.

If $j+3r+2\neq n$, then $|N_{G_{j+3r+2}}(v_{j+3r+2})|\geq 2$. If $N_{G_{j+3r+2}}[v_{j+3r+2}]\cap D'=N_{G_{j+3r+2}}[v_{j+3r+2}]\cap D=\{v_{j+3r+2}\}$ i.e., $N_{G_{j+3r+2}}(v_{j+3r+2})\cap D=\emptyset$, then let $S_4=\{v_{j+3},v_{j+6},\ldots,v_{j+3r},v_{j+3r+3}\}$ and $S_5=D'\cap S_4$. Notice that $|S_4|=|S_1|$. Now let $D_2=(D'\setminus S_1)\cup (S_4\setminus S_5)$. If $S_5\neq \emptyset$, then $D_2$ is a \ntd-set of $G$ with $|D_2|<|D|$. This is again a contradiction to the minimality of $D$. So $S_5=\emptyset$ and hence $D_2$ is a minimum \ntd-set of $G$ containing $v_j$. If $N_{G_{j+3r+2}}[v_{j+3r+2}]\cap D'=N_{G_{j+3r+2}}[v_{j+3r+2}]\cap D\neq \{v_{j+3r+2}\}$, then let $S_6=\{v_{j+2},v_{j+5},\ldots,v_{j+3r-1}\}$, $S_7=\{v_{j+3},v_{j+6},\ldots,v_{j+3r}\}$, and $S_8=D'\cap S_7$. Notice that $|S_6|=|S_7|$. Let $D_3=(D'\setminus S_6)\cup (S_7\setminus S_8)$. Recall that $v_{j+3r+2}\in D'$. If $S_8\neq \emptyset$, then $D_3$ is a \ntd-set of $G$ with $|D_3|<|D|$. This is again a contradiction to the minimality of $D$. So $S_8=\emptyset$ and hence $D_3$ is a minimum \ntd-set of $G$ containing $v_j$. This completes the proof of the claim. 
\end{proof}

We now return to the proof of Lemma \ref{pig2}. If $|N_{G_j}(v_j)|\neq 1$, then by Claim \ref{pig2claim1.1}, no vertex of $N_{G_j}(v_j)$ is an isolated vertex in $G[N_G(D')]$. This implies that $D'$ is a minimum \ntd-set of $G$ containing $v_j$. If $|N_{G_j}(v_j)|= 1$, then let $p$ be the minimum index such that $p>j$, $G[\{v_j,v_{j+1},\ldots,v_p\}]$ is a path, and $|N_{G_p}(v_p)|\neq 1$. If $G[\{v_j,v_{j+1},\ldots,v_p\}]$ is a path of length $3r$ for some $r\geq 1$, then by Claim \ref{pig2claim1.5}, either no vertex of $N_{G_j}(v_j)$ is an isolated vertex in $G[N_G(D')]$ or there is a minimum \ntd-set of $G$ containing $v_j$.
If $G[\{v_j,v_{j+1},\ldots,v_p\}]$ is a path of length $3r+1$ (resp. $3r+2$) for some $r\geq 0$, then by Claim \ref{pig2claim1.6} (resp. Claim \ref{pig2claim1.7}), either no vertex of $N_{G_j}(v_j)$ is an isolated vertex in $G[N_G(D')]$ or there is a minimum \ntd-set of $G$ containing $v_j$. This completes the proof of the lemma.
\end{proof}

\begin{lemma}\label{pig1}
Let $G$ be a connected proper interval graph having at least three vertices and $\sigma=(v_1,v_2,\ldots,v_n)$ be a BCO of $G$. If $v_1$ is a pendant vertex of $G$, then the following are true.
\begin{enumerate}
\item If $\ell(v_{\ell(v_2)})=\ell(v_{\ell(v_3)})$, then there is a minimum \ntd-set of $G$ containing $\{v_2,v_{\ell(v_2)}\}$.
\item If $\ell(v_{\ell(v_2)})\neq \ell(v_{\ell(v_3)})$, then there is a minimum \ntd-set of $G$ containing $\{v_1,v_{\ell(v_3)}\}$.
\end{enumerate}
\end{lemma}

\begin{proof}
Let $D$ be a minimum \ntd-set of $G$ and $k$ be the minimum index with respect to $\sigma$ such that $v_k\in D$ and $v_k$ dominates $v_1$.

(a) Since $\sigma$ is a BCO of $G$ and $v_1$ is a pendant vertex of $G$, $N_G[v_1]=\{v_1,v_2\}$. Depending on whether $v_2\in D$ or $v_2\notin D$, we consider the following two cases and in each case either we show that $v_2,v_{\ell(v_2)}\in D$ or construct a minimum \ntd-set $D'$ of $G$ from $D$ such that $v_2,v_{\ell(v_2)}\in D'$.

\noindent\textbf{Case 1:} $v_2\in D$.

If $v_{\ell(v_2)}\in D$, then we are done. So assume that $v_{\ell(v_2)}\notin D$.
If $k=1$, then let $D'=(D\setminus\{v_1\})\cup\{v_{\ell(v_2)}\}$. Clearly $D'$ is a minimum \ntd-set of $G$ since $D$ is a minimum \ntd-set of $G$. If $k\neq 1$, then $k=2$ since $v_2\in D$. Since $v_1$ is a pendant vertex of $G$ and $v_1\notin D$, by Observation \ref{obsprelim0}, $|N_G[v_2]\cap D|\geq 2$. So there must be a vertex $v_r$ such that $r\neq 1$ and $v_r\in N_G(v_2)\cap D$. Moreover, since $v_{\ell(v_2)}\notin D$, $r\neq \ell(v_2)$. So $r<\ell(v_2)$. Now let $D'=(D\setminus\{v_r\})\cup \{v_{\ell(v_2)}\}$. We now prove that $D'$ is also a \ntd-set of $G$. Since $\sigma$ is a BCO of $G$, by Observation \ref{obsprelim2}(a), $\{v_2,v_3,\ldots,v_{\ell(v_2)}\}$ is a clique and hence $v_rv_{\ell(v_2)}\in E(G)$. Moreover, $v_r$ is a simplicial vertex in $G_r=G[\{v_r,v_{r+1},\ldots,v_n\}]$. So by Observation \ref{obsprelim5}(b), $N_{G_r}[v_r]\subseteq N_{G_r}[v_{\ell(v_2)}]$. This implies that $N_{G}[v_r]\subseteq N_G[v_{\ell(v_2)}]$ since $v_1$ is not adjacent to $v_r$. So $D'$ is a dominating set of $G$. To prove $D'$ is a \ntd-set, we need to show that no vertex from $N_G[v_2]\cup N_{G}[v_{\ell(v_2)}]$ is an isolated vertex in $G[N_G(D')]$. Since $v_2,v_{\ell(v_2)}\in D'$, no vertex from$N_G[v_2]\cup N_{G}[v_{\ell(v_2)}]$ is an isolated vertex in $G[N_G(D')]$. This implies that $D'$ is also a minimum \ntd-set of $G$ since $|D'|=|D|$.

\noindent\textbf{Case 2:} $v_2\notin D$.

We note that $k=1$ since $v_1$ is a pendant vertex of $G$ and $v_2\notin D$. Moreover, since $v_3v_1\notin E(G)$, there must be a vertex in $D$ different from $v_2$ that dominates $v_3$. Among all such vertices, let $v_r$ be the vertex with the minimum index. Clearly $3\leq r\leq \ell(v_3)$. If $v_r$ is adjacent to $v_2$, then let $D'=(D\setminus\{v_1\})\cup \{v_{2}\}$. Notice that $D'$ is a minimum \ntd-set of $G$. So by Case 1, we can find a minimum \ntd-set of $G$ containing $v_2$ and $v_{\ell(v_2)}$. Now assume that $v_2$ is not adjacent to $v_r$. Let $D'=(D\setminus\{v_k,v_r\})\cup \{v_2,v_{\ell(v_2)}\}$. Since $v_2v_r\notin E(G)$ and $v_3v_r\in E(G)$, $\ell(v_2)<r\leq \ell(v_3)$. This implies that $3\leq \ell(v_2)<r\leq \ell(v_3)$. So by Observation \ref{obsprelim2}(a), $\{v_3,v_4,\ldots,v_{\ell(v_3)}\}$ is a clique and hence $\{v_3,v_4,\ldots,v_{\ell(v_3)}\}\subseteq N_G[v_r]$. Also we notice that $\{v_3,v_4,\ldots,v_{\ell(v_3)}\}\subseteq N_G[v_{\ell(v_2)}]$ and $v_rv_{\ell(v_2)}\in E(G)$. So $N_G[v_r]\subseteq N_G[v_{\ell(v_2)}]$ since $\ell(v_{\ell(v_2)})=\ell(v_{\ell(v_3)})$. This implies that $D'$ is a dominating set of $G$.  Since $v_2,v_{\ell(v_2)}\in D'$, no vertex from $N_G[v_2]\cup N_{G}[v_{\ell(v_2)}]$ is an isolated vertex in $G[N_G(D')]$. So $D'$ is also a minimum \ntd-set of $G$ since $|D'|=|D|$. This completes the proof of (a).

(b) Recall that $v_k\in D$ is the minimum indexed vertex that dominates $v_1$.

\begin{claim}\label{pig1claim1}
If $k=1$, then there is a minimum \ntd-set of $G$ containing $\{v_1,v_{\ell(v_3)}\}$.
\end{claim}
\begin{proof}[Proof of Claim \ref{pig1claim1}:]
If $v_{\ell(v_3)}\in D$, then we are done. So assume that $v_{\ell(v_3)}\notin D$. To dominate $v_3$, there must be a vertex $v_a\in N_G[v_3]$ in $D$. Now let $D'=(D\setminus\{v_a\})\cup \{v_{\ell(v_3)}\}$. Since $3\leq a< \ell(v_3)$, by Observation \ref{obsprelim2}(a), $\{v_3,v_4,\ldots,v_{\ell(v_3)}\}$ is a clique and hence $N_G[v_a]\cap \{v_3,v_4,\ldots,v_{\ell(v_3)}\}\subseteq N_G[v_{\ell(v_3)}]$. Again since $v_a$ is a simplicial vertex in $G_a=G[\{v_a,v_{a+1},\ldots,v_n\}]$, by Observation \ref{obsprelim5}(b), $N_{G_a}[v_a]\subseteq N_{G_a}[v_{\ell(v_3)}]$. So $D'$ is a dominating set of $G$.  Moreover, no vertex from the set $\{v_1,v_2,\ldots,v_{\ell(v_3)}\}$ is an isolated vertex in $G[N_G(D')]$. So if there is an isolated vertex in $G[N_G(D')]$, then it must belong to $N_G(v_{\ell(v_3)})\cap \{v_{\ell(v_3)+1}, v_{\ell(v_3)+2},\ldots, v_n\}$. Let $j=\ell(v_3)$. Now considering the cardinality of $N_{G_j}(v_j)$ similar to that we have done in Lemma \ref{pig2}, we can show that either $D'$ is a minimum \ntd-set of $G$ containing $\{v_1,v_{\ell(v_3)}\}$ or we can construct a minimum \ntd-set of $G$ from $D'$ that contains $\{v_1,v_{\ell(v_3)}\}$. This completes the proof of the claim.
\end{proof}

If $k\neq 1$, then $v_2,v_r\in D$, where $v_r$ is a neighbor of $v_2$. This is true since $v_1$ is a pendant vertex of $G$. Clearly $3\leq r\leq \ell(v_2)$.

If $v_{\ell(v_3)}\in D$, then let $D_1=(D\setminus\{v_r,v_2\})\cup \{v_1\}$. Since $\ell(v_r)\leq \ell(v_2)$ and $\ell(v_2)<\ell(v_3)$, by Observation \ref{obspig1-2}, we have $\ell(v_{\ell(v_r)})\leq \ell(v_{\ell(v_2)})\leq \ell(v_{\ell(v_3)})$. So no vertex of $\{v_2,v_3,\ldots,v_{p}\}$, where $p=\ell(v_{\ell(v_3)})$, is an isolated vertex in $G[N_G(D_1)]$. Since $v_1$ dominates $v_2$ and $v_{\ell(v_3)}$ dominates the vertex set $\{v_3, v_4,\ldots,v_{p}\}$, $D_1$ is a  \ntd-set of $G$ smaller than $D$. This is a contradiction. So $v_{\ell(v_3)}\notin D$. Let $D_2=(D\setminus\{v_2\})\cup \{v_1\}$. Since $\sigma$ is a BCO of $G$ and $3\leq r\leq \ell(v_2)$, $N_{G}[v_2]\cap \{v_2,v_3\ldots,v_n\}\subseteq N_{G}[v_r]\cap  \{v_2,v_3\ldots,v_n\}$. So $D_2$ is a dominating set of $G$. Moreover, there is no isolated vertex in $G[N_G(D_2)]$. So $D_2$ is also a minimum \ntd-set of $G$ since $|D_2|=|D|$. Again since $v_1$ is a pendant vertex of $G$ and $v_2\notin D_2$, $v_1$ is the minimum indexed vertex present in $D_2$ that dominates $v_1$. So by Claim \ref{pig1claim1}, there is a minimum \ntd-set of $G$ containing $\{v_1,v_{\ell(v_3)}\}$. This completes the proof of (b) and thus the proof of the lemma.
\end{proof}

We now describe our algorithm, namely \textsc{MNTDS-PIG($G$)} to compute a minimum \ntd-set of a given connected proper interval graph $G$. To have a better understanding of the algorithm \textsc{MNTDS-PIG($G$)}, we first explain the different cases of the algorithm and the updates that will be made. The details of the algorithm \textsc{MNTDS-PIG($G$)} are as follows. If $G$ is a proper interval graph with at most two vertices, then by Observation \ref{obsprelim1}, it is easy to construct a minimum \ntd-set of $G$. Therefore, below we discuss the details of \textsc{MNTDS-PIG($G$)} for a given connected proper interval graph with at least $3$ vertices. Let $G$ be a connected proper interval graph having vertices $v_1,v_2\ldots,v_n; n\geq 3$.

\begin{itemize}
\item \textsc{MNTDS-PIG($G$)} first computes a BCO $\sigma=(v_1,v_2,\ldots,v_n)$ of $G$ and maintains a set $D$ i.e., a set that will form a minimum \ntd-set of $G$.

\item \textsc{MNTDS-PIG($G$)} checks the status (dominated or undominated with respect to the set $D$) of every vertex $v_i$; $1\leq i\leq n$ one by one with respect to $\sigma$ and decides the appropriate vertex or vertices to be selected into $D$.

\item In the first step i.e., while considering the vertex $v_1$, if $v_1$ is a pendant vertex of $G$, then \textsc{MNTDS-PIG($G$)} selects the vertices into $D$ according to Lemma \ref{pig1}.

\begin{itemize}
\item[$\bullet$] If \textsc{MNTDS-PIG($G$)} selects the vertices into $D$ according to Lemma \ref{pig1}(a), then notice that no vertex from $\{v_1,v_2,\ldots, v_{\ell(v_2)}\}$ is an isolated vertex in $G[N_G(D)]$, where $D$ is the updated set after first iteration. So $i=1$ is updated to $i=\ell(v_{\ell(v_2)})+1$.

\item[$\bullet$] Suppose \textsc{MNTDS-PIG($G$)} selects the vertices into $D$ according to Lemma \ref{pig1}(b). If $|N_G(v_{\ell(v_3)})\cap V(G_{\ell(v_3)+1})|\geq 2$ or $v_{\ell(v_3)+1}v_{\ell(v_3)-1}\in E(G)$, then notice that no vertex from $\{v_1,v_2,\ldots, v_{v_{\ell(\ell(v_3))}}\}$ is an isolated vertex in $G[N_G(D)]$, where $D$ is the updated set after first iteration. Otherwise notice that no vertex from $\{v_1,v_2,\ldots,v_{\ell(v_3)}\}$ is an isolated vertex in $G[N_G(D)]$. So $i=1$ is updated to either $i=\ell(v_{\ell(v_3)})+1$ or $i=\ell(v_{\ell(v_3)})$.
\end{itemize}
\item If $v_1$ is not a pendant vertex of $G$, then \textsc{MNTDS-PIG($G$)} selects $v_j$, where $j=\ell(v_1)$ according to Lemma \ref{pig2}. If $|N_G(v_j)\cap V(G_{j+1})|\geq 2$ or $v_{j+1}v_{j-1}\in E(G)$, then notice that no vertex from $\{v_1,v_2,\ldots,v_{\ell(v_j)}\}$ is an isolated vertex in $G[N_G(D)]$, where $D$ is the updated set after first iteration. Otherwise notice that no vertex from $\{v_1,v_2,\ldots,v_j\}$ is an isolated vertex in $G[N_G(D)]$. So $i=1$ is updated to either $i=\ell(v_j)+1$ or $i=\ell(v_j)=j+1$.

\item Suppose \textsc{MNTDS-PIG($G$)} is considering the vertex $v_i$, where $i>1$.

\begin{itemize}
\bitem If $v_i$ is not dominated by the so far constructed set $D$, then \textsc{MNTDS-PIG($G$)} selects the vertex $v_j$, where $j=\ell(v_i)$ and $i$ is updated to either $\ell(v_j)+1$ or $\ell(v_j)$.

\bitem If $v_i$ is dominated, then \textsc{MNTDS-PIG($G$)} selects new vertices to $D$ considering two cases. First one is either $i=n$ or $i+1=n$ and the second one is $i\neq n$ and $i+1\neq n$. In the first case, since $v_i$ is selected, no vertex is an  isolated vertex in $G[N_G(D)]$. If $i\neq n$ and $i+1\neq n$, then \textsc{MNTDS-PIG($G$)} avoids to select $v_{i+1}$ so that $v_i,v_{i+1}\in V\setminus D$ and selects $v_p$, where $p=\ell(v_{i+1})$. Notice that $\ell(v_{i+1})\neq n$. Then $i$ is updated to $\ell(v_{p})+1$ if $|N_G(v_p)\cap V(G_{p+1})|\geq 2$ or $v_{p+1}v_{p-1}\in E(G)$; otherwise $i$ is updated to $\ell(v_p)=p+1$.
\end{itemize}

\end{itemize}

\begin{algorithm}[H]
\SetAlgoVlined
{\relsize{-1}
\KwIn{A connected proper interval graph $G=(V ,E)$ with at least three vertices\;}
\KwOut{A minimum \ntd-set $D$ of $G$\;}
Compute a BCO $\sigma=(v_1, v_2, \ldots, v_n)$ of $G$\;
Initialize $D=\emptyset$\;

$i=1$\;
\If{$(v_1$ is a pendant vertex of $G)$}{\eIf(\tcc*[f]{Case 1}){$(\ell(v_{\ell(v_2)})=\ell(v_{\ell(v_{3})}))$}{$D=D\cup \{v_2, v_{\ell(v_2)}\}$\;$i=\ell(v_{\ell(v_2)})+1$\;}(\tcc*[f]{Case 2}){$D\cup \{v_1, v_{\ell(v_{3})}\}$\;
$
i=
\begin{cases}
\ell(v_{\ell(v_{3})})+1,& \text{if $|N_G(v_{\ell(v_{3})})\cap V(G_{\ell(v_{3})+1})|\geq 2$ or}\\
 & \text{ $v_{\ell(v_{3})+1} v_{\ell(v_{3})-1}\in E$\;}\\
\ell(v_{\ell(v_{3})}),& \text{otherwise.}
\end{cases}
$}}
\While{$(i\leq n)$}{ Let $\ell(v_i)=j$\;\eIf(\tcc*[f]{Case 3}){$(v_i$ is not dominated$)$}{$D=D\cup \{v_j\}$\;
$
i=
\begin{cases}
\ell(v_j)+1,& \text{if $|N_G(v_j)\cap V(G_{j+1})|\geq 2$ or  $v_{j-1} v_{j+1}\in E$\;}\\
\ell(v_j),& \text{otherwise.}
\end{cases}
$
}
{
 \eIf(\tcc*[f]{Case 4}){$((d_{G_i}(v_i)=0)$ or $(d_{G_i}(v_i)=1$ and $v_j=v_n) )$}{$D=D\cup \{v_i\}$\; \Return $D$\;}(\tcc*[f]{Case 5}){Let $\ell(v_{i+1})=p$\;$D=D\cup \{v_p\}$\;$
i=
\begin{cases}
\ell(v_{p})+1,& \text{if $|N_G(v_{p})\cap V(G_{p+1})|\geq 2$ or  $v_{p-1} v_{p+1}\in E$\;}\\
\ell(v_{p}),& \text{otherwise.}
\end{cases}
$}

 }
}

\Return $D$\;

\caption{\label{ntds-pig2}\textsc{MNTDS-PIG($G$)}}}
\end{algorithm}


Before presenting the proof of correctness of the algorithm \textsc{MNTDS-PIG($G$)}, we illustrate different iterations of the algorithm \textsc{MNTDS-PIG($G$)} through the graphs shown in Figure \ref{algoillu}. We have taken two proper interval graphs in Figure \ref{algoillu} to explain the two cases of Lemma \ref{pig1} since both the cases cannot occur simultaneously. Notice that in Figure \ref{algoillu}, $\sigma_1=(v_1,v_2,\ldots,v_{11})$ is a BCO of $\mathcal{H}_1$ and $\sigma_2=(v_1,v_2,\ldots,v_{15})$ is a BCO of $\mathcal{H}_2$. We maintain two sets $D_1$ and $D_2$ that will be computed by the algorithm \textsc{MNTDS-PIG($G$)} on $\mathcal{H}_1$ and $\mathcal{H}_2$, respectively. Let $D_1=D_2=\emptyset$. Notice that $\{v_2,v_4,v_9,v_{10}\}$ is a minimum \ntd-set of $\mathcal{H}_1$ and $\{v_1,v_5,v_9,v_{13},v_{14}\}$ is a minimum \ntd-set of $\mathcal{H}_2$.

\begin{figure}[h]
\begin{center}
\includegraphics[scale=.4]{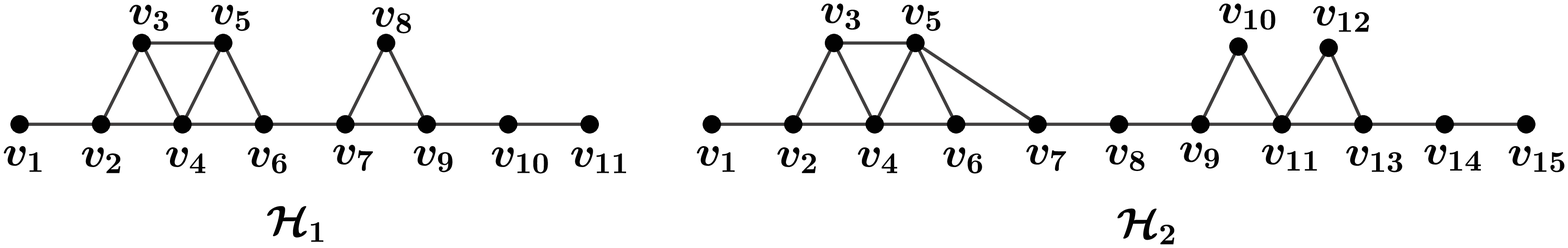}
\caption{\label{algoillu} Two proper interval graphs $\mathcal{H}_1$ and $\mathcal{H}_2$}
\end{center}
\end{figure}

\noindent\textbf{Illustration of the algorithm \textsc{MNTDS-PIG($G$)}:}
\begin{itemize}
\item In both $\mathcal{H}_1$ and $\mathcal{H}_2$, $v_1$ is a pendant vertex. So Case 1 or Case 2 of the algorithm will be applied at the first iteration of the algorithm.

\item \textbf{1st iteration:} In $\mathcal{H}_1$, $\ell(v_2)=4$, $\ell(v_3)=5$, and $\ell(v_4)=\ell(v_5)=6$. So Case 1 of the algorithm \textsc{MNTDS-PIG($G$)} is applied and $D_1=D_1\cup\{v_2,v_4\}$. Then the index of the vertex that will be considered in the next iteration of the algorithm is updated to $\ell(v_4)+1=7$ (see Line 7).

In $\mathcal{H}_2$,  $\ell(v_2)=4$, $\ell(v_3)=5$, $\ell(v_4)=6$, and $\ell(v_5)=7$. Thus $\ell(v_4)\neq \ell(v_5)$. So Case 2 of the algorithm \textsc{MNTDS-PIG($G$)} is applied and $D_2=D_2\cup\{v_1,v_5\}$. Then the index of the vertex that will be considered in the next iteration of the algorithm is updated to $\ell(v_5)+1=8$ (see Line 10).

\item \textbf{2nd iteration:} In $\mathcal{H}_1$, the second iteration of the algorithm \textsc{MNTDS-PIG($G$)}  considers the vertex $v_7$. Notice that $v_7$ is not dominated by $D_1$. So Case 3 of the algorithm \textsc{MNTDS-PIG($G$)} is applied and $D_1=D_1\cup\{v_9\}$. Then the index of the vertex that will be considered in the next iteration of the algorithm is updated to $\ell(v_7)+1=10$ (see Line 15).

In $\mathcal{H}_2$, the second iteration of the algorithm \textsc{MNTDS-PIG($G$)}  considers the vertex $v_8$. Notice that $v_8$ is not dominated by $D_2$. So Case 3 of the algorithm \textsc{MNTDS-PIG($G$)} is applied and $D_2=D_2\cup\{v_9\}$. Then the index of the vertex that will be considered in the next iteration of the algorithm is updated to $\ell(v_8)+1=12$ (see Line 15).

\item \textbf{3rd iteration:} In $\mathcal{H}_1$, the third iteration of the algorithm \textsc{MNTDS-PIG($G$)}  considers the vertex $v_{10}$. Notice that $v_{10}$ is dominated by $D_1$, $d_{G_{11}}(v_{10})=1$, and $\ell(v_{10})=11$. So Case 4 of the algorithm \textsc{MNTDS-PIG($G$)} is applied and $D_1=D_1\cup\{v_{10}\}$.

In $\mathcal{H}_2$, the third iteration of the algorithm \textsc{MNTDS-PIG($G$)}  considers the vertex $v_{12}$. Notice that $v_{12}$ is not dominated by $D_2$. So Case 3 of the algorithm \textsc{MNTDS-PIG($G$)} is applied and $D_2=D_2\cup\{v_{13}\}$. Then the index of the vertex that will be considered in the next iteration of the algorithm is updated to $\ell(v_{12})+1=14$ (see Line 15).

\item \textbf{4th iteration:} In $\mathcal{H}_2$, the fourth iteration of the algorithm \textsc{MNTDS-PIG($G$)}  considers the vertex $v_{14}$. Notice that $v_{14}$ is dominated by $D_2$, $d_{G_{14}}(v_{14})=1$, and $\ell(v_{14})=15$. So Case 4 of the algorithm \textsc{MNTDS-PIG($G$)} is applied and $D_2=D_2\cup\{v_{14}\}$.
\end{itemize}

Suppose that \textsc{MNTDS-PIG($G$)} executes for $k$ number of iterations. Then $k\leq n$. Let $D_r$, $1\leq r\leq k$ be the set constructed by \textsc{MNTDS-PIG($G$)} after the execution of the $r$-th iteration. The following can be observed from the algorithm  \textsc{MNTDS-PIG($G$)}. However, to have better understanding, we have given the proofs for Observation \ref{obspig2}-\ref{obspig4}.

\begin{obs}\label{obspig2-0}
For every $0\leq r<k$, if $v_a$ is not an isolated vertex in $G[N_G(D_r)]$, then $v_a$ is not an isolated vertex in $G[N_G(D_{r'})]$ for every $r<r'\leq k$.
\end{obs}

Notice that if a vertex $v_a$ is selected by the algorithm \textsc{MNTDS-PIG($G$)} at some iteration, then the next iteration of the algorithm \textsc{MNTDS-PIG($G$)} considers the vertex $v_b$, where $b=a+1=\ell(v_a)$ (see Line 10, 15, and 22 of the algorithm) or $b=\ell(v_a)+1$. To use this fact later, we record this in the following observation. 

\begin{obs}\label{obspig2-1}
If $v_a$ is the vertex selected by the algorithm \textsc{MNTDS-PIG($G$)} at some iteration, then the next iteration of the algorithm considers the vertex $v_b$, where $b\geq \ell(v_a)$.
\end{obs}

We now present two lemmas that will help in proving that the set $D_k$ is a \ntd-set of $G$.

\begin{lemma}\label{obspig2}
Let $v_i$ be the vertex of $G$ considered by the algorithm \textsc{MNTDS-PIG($G$)} at the beginning of the $r$-th iteration. If $v_i$ is dominated by $D_{r-1}$, then  $|N_G(v_{i})\cap \{v_1,v_2,\ldots,v_{i-1}\}|=1$, $v_{i-1}\in D_{r-1}$,  $N_G(v_{i-1})\cap D_{r-1}=\emptyset$, and $v_{i-2}v_{i},v_{i-1}v_{i+1}\notin E(G)$.
\end{lemma}
\begin{proof} Let $\sigma=(v_1,v_2,\ldots,v_n)$ be a BCO of $G$. Since $D=\emptyset$ at the beginning of the 1st iteration, $v_1$ cannot be dominated. So assume that $r\neq 1$ and $i\neq 1$. Since $\sigma$ is a BCO of $G$, $v_i$ is a simplicial vertex in $G[\{v_1,v_2,\ldots,v_{i}\}]$ as well as $G[\{v_i,v_{i+1},\ldots,v_n\}]$. Let $v_a$ be the maximum indexed vertex selected at the $(r-1)$-th iteration. So $v_a\in D_{r-1}$. Since the algorithm considers the vertex $v_i$ at the $r$-th iteration and $v_i$ is dominated by $D_{r-1}$, by Observation \ref{obspig2-1}, $\ell(v_a)=i$. Since $\sigma$ is a BCO of $G$, $v_av_i\in E(G)$.  Moreover, since $v_i$ is the considered at the $r$-th iteration, $a\neq i$. This implies that $a<i$. 

\begin{claim}\label{obspig2claim1}
$|N_G(v_i)\cap \{v_1,v_{2},\ldots,v_{i-1}\}|=1$.

\end{claim}
\begin{proof}[Proof of Claim \ref{obspig2claim1}:]
If $|N_G(v_i)\cap \{v_1,v_{2},\ldots,v_{i-1}\}|\geq 2$, then there must exist a vertex $v_b$ other than $v_a$. Since $v_i$ is simplicial in $G[\{v_1,v_2,\ldots,v_i\}]$, $v_av_b\in E(G)$. Notice that $b<i$. If $a<b$, then $|N_G(v_a)\cap V(G_{a+1})|\geq 2$. Again if $a>b$, then since $\{v_b,\ldots ,v_i\}$ is a clique, $v_{a-1}v_{a+1}\in E(G)$. In both cases, the next iteration of the algorithm would have considered the vertex $v_c$, where $c=\ell(v_a)+1$. Recall that $\ell(v_a)=i$. So we have $\ell(v_a)+1>i$ and hence $c>i$. This is a contradiction to the fact that the vertex $v_i$ is considered at $r$-th iteration of the algorithm. So $|N_G(v_i)\cap \{v_1,v_{2},\ldots,v_{i-1}\}|\leq 1$. This implies that $|N_G(v_i)\cap \{v_1,v_{2},\ldots,v_{i-1}\}|= 1$ since $v_a\in N_G(v_i)$.
\end{proof}
\begin{claim}\label{obspig2claim2}
$N_G(v_{i-1})\cap D_{r-1}\neq \emptyset$ and $v_{i-2}v_{i},v_{i-1}v_{i+1}\notin E(G)$.
\end{claim}
\begin{proof}[Proof of Claim \ref{obspig2claim2}:]
Since $v_i$ is dominated by $D_{r-1}$ and $v_av_i\in E(G)$ with $a<i$, by Claim \ref{obspig2claim1}, it is clear that $v_{i-1}\in D_{r-1}$. If $N_G(v_{i-1})\cap D_{r-1}\neq \emptyset$, then, since $v_{i-1}\in D_{r-1}$, by Observation \ref{obspig2-1}, the next iteration of the algorithm would have considered the vertex $v_c$, where $c\geq \ell(v_{i-1})$. In particular, $c=\ell(v_{i-1})+1$ as $v_{\ell(v_{i-1})}$ is not an isolated vertex in $G[N_G(D_{r-1})]$. Since $v_{i-1}v_i\in E(G)$, $\ell(v_{i-1})+1>i$. This is a contradiction to the fact that the vertex $v_i$ is considered at $r$-th iteration of the algorithm. So $N_G(v_{i-1})\cap D_{r-1}=\emptyset$. If $v_{i-2}v_{i}\in E(G)$ or $v_{i-1}v_{i+1}\in E(G)$, then the next iteration of the algorithm would have considered the vertex $v_c$, where $c=\ell(v_{i-1})+1$ (it will consider the vertex $v_c$ using Line 10 or Line 15 or Line 22 of the algorithm). Since $\ell(v_{i-1})+1>i$, this is again a contradiction to the fact that the vertex $v_i$ is considered at $r$-th iteration of the algorithm. So $v_{i-2}v_{i},v_{i-1}v_{i+1}\notin E(G)$. 
\end{proof}
By Claim \ref{obspig2claim1} and Claim \ref{obspig2claim2}, the proof of the lemma follows. 
\end{proof}

\begin{lemma}\label{obspig2-01}
Let $v_a$ and $v_b$ be the vertices considered by the algorithm \textsc{MNTDS-PIG($G$)} at the beginning of the $s$-th and the $(s+1)$-th iterations, respectively. Then every $x\in \{v_a,v_{a+1},\ldots,v_{b-1}\}$ is dominated by $D_{s}$ and is not an isolated vertex in $G[N_G(D_{s})]$.
\end{lemma}
\begin{proof}
If $a=1$, then the algorithm chooses vertices using Lemma \ref{pig2} or Lemma \ref{pig1} and hence the vertex $v_a$ is dominated. If $v_1$ is a pendant vertex of $G$, then $s=1$ and $D_1=\{v_2,v_{\ell(v_2)}\}$ or $D_1=\{v_1,v_{\ell(v_3)}\}$. If $D_1=\{v_2,v_{\ell(v_2)}\}$, then since the next iteration of the algorithm considers the vertex $v_b$, $b=\ell(v_{\ell(v_2)})+1$. Since $v_2v_{\ell(v_2)}\in E(G)$, every vertex $x\in \{v_1,v_2,\ldots,v_{b-1}\}\}$ is dominated by $D_1$ and $x$ is not an isolated vertex in $G[N_G(D_1)]$. If $D_1=\{v_1,v_{\ell(v_3)}\}$, then $b=\ell(v_3)+1$ or $b=\ell(v_{\ell(v_3)})+1$. Notice that in each case, every vertex $x\in \{v_1,v_2,\ldots,v_{b-1}\}$ is dominated by $D_1$ and $x$ is not an isolated vertex in $G[N_G(D_1)]$. Similarly if $v_1$ is not a pendant vertex of $G$, then $D_1=\{v_{j}\}$, where $\ell(v_1)=j$. Then $b=\ell(v_j)+1$ or $b=j+1$. In each case, every vertex $x\in \{v_1,v_2,\ldots,v_{b-1}\}$ is dominated by $D_1$ and $x$ is not an isolated vertex in $G[N_G(D_1)]$.

First assume that $a\neq 1$. Clearly $a<b$. If $v_a$ is not dominated by $D_{s-1}$, then at the $s$-th iteration the algorithm would have chosen the vertex $v_c$, where $c=\ell(v_a)$. Since $a<b$ and the vertex $v_b$ is considered at $(s+1)$-th iteration, we have $a\neq n$ and $\ell(v_a)\neq n$. This implies that $a<\ell(v_a)$. Moreover, since the next iteration of the algorithm considers the vertex $v_b$, we have $b=c+1$ or $b=\ell(v_c)+1$. Notice that in each case, every vertex $x\in \{v_a,v_{a+1},\ldots,v_{b-1}\}$ is dominated by $D_s$ and $x$ is not an isolated vertex in $G[N_G(D_s)]$.

Now assume that $a\neq 1$ and $v_a$ is dominated by $D_{s-1}$. Clearly $a\neq n$ since $a<b$. If $a=n-1$, then the algorithm would have selected the vertex $v_a$ (see Line 19) at the $s$-th iteration. This implies that $v_a$ is not an isolated vertex in $G[N_G(D_s)]$. So assume that $a\neq n-1$. Then by Lemma \ref{obspig2}, $|N_G(v_a)\cap \{v_1,v_2,\ldots,v_{i-1}\}|=1$, $v_{a-1}\in D_{s-1}$, $N_G(v_{a-1})\cap D_{s-1}=\emptyset$, and $v_{a-2}v_{a},v_{a-1}v_{a+1}\notin E(G)$. Then the algorithm would have selected the vertex $v_{a'}$, where $a'=\ell(v_{a+1})$ (see Line 22) at the $s$-th iteration. Since $v_{a+1}\notin D_s$, $v_a$ is not an isolated vertex in $G[N_G(D_s)]$. Since the next iteration of the algorithm considers the vertex $v_b$, either $b=a'+1$ or $b=\ell(v_{a'})+1$. Notice that in each case, every vertex $x\in \{v_a,v_{a+1},\ldots,v_{b-1}\}$ is dominated by $D_s$ and $x$ is not an isolated vertex in $G[N_G(D_s)]$.
\end{proof}

Using Observation \ref{obspig2-0} and Lemma \ref{obspig2-01}, we have the following corollary.

\begin{coro}\label{obspig4}
If $v_i$ is the vertex of $G$ considered by the algorithm \textsc{MNTDS-PIG($G$)} at the beginning of the $r$-th iteration, then every $x\in \{v_1,v_2,\ldots,v_{i-1}\}$ is dominated by $D_{r-1}$ and is not an isolated vertex in $G[N_G(D_{r-1})]$.
\end{coro}

By Corollary \ref{obspig4}, at the end of the $k$-th iteration, every vertex $x$ of the graph $G$ is dominated by $D_k$ and is not an isolated vertex in $G[N_G(D_k)]$. Therefore, $D_k$ is a \ntd-set of $G$. We record this in the following lemma.

\begin{lemma}\label{pigcorrect1}
$D_k$ is a \ntd-set of the connected proper interval graph $G$.
\end{lemma}

Next we prove that $D_k$ is a minimum \ntd-set of $G$. To prove this, we use induction on the number of iterations. In particular, we prove that for each iteration $r$, $1\leq r\leq k$, the set $D_r$ is contained is some minimum \ntd-set of $G$.

\begin{lemma}\label{pigcorrect2}
For each $r$, $1\leq r\leq k$, $D_r$ is contained in some minimum \ntd-set of the connected proper interval graph $G$.
\end{lemma}
\begin{proof}
We prove this by induction on $r$. Let $r=1$. If $d_G(v_1)\geq 2$, then $D_1=\{v_{\ell(v_1)}\}$ and by Lemma \ref{pig1}, we conclude that there is a minimum \ntd-set of $G$ containing $D_1$. If $v_1$ is a pendant vertex of $G$, then $D_1=\{v_2,v_{\ell(v_2)}\}$ or $D_1=\{v_1,v_{\ell(v_3)}\}$. Then by Lemma \ref{pig2}, we conclude that there is a minimum \ntd-set of $G$ containing $D_1$. So the base case of the induction is true. Now assume that the induction hypothesis is true for $s=r-1$, i.e., $D_{r-1}$ is contained in some minimum \ntd-set, say $D^*$ of $G$.

Let $v_i$ be the vertex considered at the $r$-th iteration of the algorithm.  We need to consider the following cases.

\noindent\textbf{Case 1:}  $v_i$ is not dominated by $D_{r-1}$ and $d_{G_i}(v_i)\geq 2$.

In this case, $D_r=D_{r-1}\cup \{v_{\ell(v_i)}\}$.

\begin{claim}\label{pigc2-claim1}
Either $(D^*\setminus D_{r-1})\subseteq \{v_i,v_{i+1},\ldots,v_n\}$ or $D^*$ can be modified to a minimum \ntd-set $D^{**}$ of $G$ such that $(D^{**}\setminus D_{r-1})\subseteq \{v_i,v_{i+1},\ldots,v_n\}$.
\end{claim}
\begin{proof}[Proof of Claim \ref{pigc2-claim1}:]
Let $S_{D^*}=\{v\in \{v_1,v_2,\ldots,v_{i-1}\}: v\in D^*\setminus D_{r-1}\}$.

\begin{subclaim}\label{pigc2-subclaim1}
$N_G[S_{D^*}]\cap V(G_i)\subseteq N_{G_i}[x]$, where $x\in N_{G_i}[v_i]$.
\end{subclaim}

\begin{proof}[Proof of Claim \ref{pigc2-subclaim1}:]
If $S_{D^*}=\emptyset$, then we are done. So assume that $S_{D^*}\neq \emptyset$. Let $v_a\in S_{D^*}$ be arbitrary. By the choice of $S_{D^*}$, we have $a<i$. So by Observation \ref{obsprelim5}(a), $N_G[v_a]\cap V(G_i)\subseteq N_{G_i}[v_i]$. Again by Observation \ref{obsprelim5}(b), $N_{G_i}[v_i]\subseteq N_{G_i}[x]$, and so $N_G[v_a]\cap V(G_i)\subseteq N_{G_i}[v_i]\subseteq N_{G_i}[x]$. Since $v_a$ is arbitrary, we have $N_G[S_{D^*}]\cap V(G_i)\subseteq N_{G_i}[x]$.
\end{proof}
If $S_{D^*}=\emptyset$, then the statement of Claim \ref{pigc2-claim1} follows. So assume that $S_{D^*}\neq \emptyset$. Let $v_a$ be the minimum indexed vertex such that $v_a\in S_{D^*}$.  

First assume that $N_{G_i}[v_i]\cap D^*=\emptyset$. If $|S_{D^*}|\geq 2$, then let $D'=(D^*\setminus S_{D^*})\cup \{v_{i+1}\}$. Since $v_{i+1}\in N_{G_i}[v_i]$, by Claim \ref{pigc2-subclaim1},  $N_G[S_{D^*}]\cap V(G_i)\subseteq N_{G_i}[v_{i+1}]$. Since $d_{G_i}(v_i)\geq 2$ and $D_{r-1}$ is a \ntd-set of $G[\{v_1,v_2,v_3,\ldots,v_{i-1}\}]$ by Corollary \ref{obspig4}, $D'$ is a smaller \ntd-set of $G$ than $D^*$. This is a contradiction. So $|S_{D^*}|=1$ and hence in particular $S_{D^*}=\{v_a\}$. Now $D^{**}=(D^*\setminus\{v_a\})\cup \{v_{i+1}\}$ is a minimum \ntd-set of $G$ such that $(D^{**}\setminus D_{r-1})\subseteq \{v_i,v_{i+1},\ldots,v_n\}$.

Now assume that $N_{G_i}[v_i]\cap D^*\neq\emptyset$. Let $v_b\in N_{G_i}[v_i]\cap D^*$. By Claim \ref{pigc2-subclaim1}, $N_G[S_{D^*}]\cap V(G_i)\subseteq  N_{G_i}[v_b]$. If $|N_{G_i}[v_i]\cap D^*|\geq 2$, then, since  $D_{r-1}$ is a \ntd-set of $G[\{v_1,v_2,\ldots,v_{i-1}\}]$ by Corollary \ref{obspig4}, $D^*\setminus S_{D^*}$ is a smaller \ntd-set of $G$. This is a contradiction. So $N_{G_i}[v_i]\cap D^*=\{v_b\}$. Now let $D'=(D^*\setminus S_{D^*})\cup\{v_{b'}\}$, where $v_{b'}\in N_{G_i}[v_i]\setminus D^*$. Since $d_{G_i}(v_i)\geq 2$, such a vertex $v_{b'}$ exists. Since $N_G[S_{D^*}]\cap V(G_i)\subseteq  N_{G_i}[v_b]$, if $|S_{D^*}|\geq 2$, then $D'$ is a smaller \ntd-set of $G$. If $|S_{D^*}|=1$, then $D'$ is a minimum \ntd-set of $G$ such that $(D^{**}\setminus D_{r-1})\subseteq \{v_i,v_{i+1},\ldots,v_n\}$.
\end{proof}

By Claim \ref{pigc2-claim1}, without loss of generality, we can assume that $D^*$ is a minimum \ntd-set of $G$ such that $(D^*\setminus D_{r-1})\subseteq \{v_i,v_{i+1},\ldots,v_n\}$. Since  $D_{r-1}$ is a \ntd-set of $G[\{v_1,v_2,v_3,\ldots,v_{i-1}\}]$ by Corollary \ref{obspig4}, $v_i$ is not dominated by $D_{r-1}$, and $d_{G_i}(v_i)\geq 2$, $D^*\setminus D_{r-1}$ must be a \ntd-set of $G_i=G[\{v_i,v_{i+1},\ldots,v_n\}]$. Again none of the vertex from $D_{r-1}$ is adjacent to any vertex of $G_i$ (this due to the fact that $v_i$ is not dominated by $D_{r-1}$). Since $G_i$ is also a proper interval graph with a BCO $\sigma'=(v_i,v_{i+1},\ldots,v_n)$, by Lemma \ref{pig2}, we can get a minimum \ntd-set $D_2$ of $G_i$ such that $v_{\ell(v_i)}\in D_2$. Let $D_3=D_{r-1}\cup D_2$. Then $D_3$ is a \ntd-set of $G$ with $|D_3|=|D_{r-1}|+|D_2|\leq |D_{r-1}|+|D^*\setminus D_{r-1}|=|D^*|$. This implies that $D_3$ is a minimum \ntd-set of $G$ containing $D_{r-1}\cup \{v_{\ell(v_i)}\}$.

\noindent\textbf{Case 2:}  $v_i$ is not dominated by $D_{r-1}$ and $d_{G_i}(v_i)\leq 1$.

In this case $D_r=D_{r-1}\cup \{v_j\}$, where $j=\ell(v_i)$. If $v_j\in D^*$, then we are done. So assume that $v_j\notin D^*$. Let $v_a\in D^*$ be the minimum indexed vertex that dominates $v_i$. Since $v_i$ is not dominated by $D_{r-1}$, $v_a\notin D_{r-1}$. Moreover, since $d_{G_i}(v_i)\leq 1$, $j=i$ or $j=i+1$.

If $j=i$, then $i=n$. Since $D_{r-1}$ is a \ntd-set of $G[\{v_1,v_2,\ldots,v_{n-1}\}]$ by Corollary \ref{obspig4} and $v_a\notin D_{r-1}$, $(D^*\setminus\{v_a\})\cup\{v_j\}$ is a minimum \ntd-set of $G$ containing $D_{r-1}\cup\{v_j\}$.

If $j=i+1$, then $a\leq i$. Now let $D'=(D^*\setminus\{v_a\})\cup \{v_j\}$. By Observation \ref{obsprelim5}, $N_G[v_a]\cap V(G_i)\subseteq N_{G_i}[v_i]\subseteq N_{G_i}[v_{j}]$. So $D'$ is a dominating set of $G$. Moreover, if there is an isolated vertex in $G[N_G(D')]$, then it is a neighbor of $v_j$ that appears after $j$. So as we have proved in the proof of Lemma \ref{pig2}, we can show considering $|N_{G_j}(v_j)|$ that either $D'$ is a minimum \ntd-set of $G$ or there exists a minimum \ntd-set of $G$ containing $D_{r-1}\cup \{v_j\}$. This is true since no vertex of $D_{r-1}$ will be removed while modifying $D'$.

\noindent\textbf{Case 3:} $v_i$ is dominated by $D_{r-1}$.

By Lemma \ref{obspig2}, $|N_G(v_i)\cap \{v_1,v_2,\ldots,v_{i-1}\}|=1$, $v_{i-1}\in D_{r-1}$, $N_G(v_{i-1})\cap D_{r-1}=\emptyset$, and $v_{i-2}v_i,v_{i-1}v_{i+1}\notin E(G)$. Since $D^*$ is a \ntd-set of $G$, $v_i$ is not an isolated vertex in $G[N_G(D^*)]$.

If $i=n$ or $i+1=n$, then in this case $D_r=D_{r-1}\cup\{v_i\}$. If $v_i\in D^*$, then we are done. So assume that $v_i\notin D^*$. If $i=n$, then, since $|N_G(v_i)\cap \{v_1,v_2,\ldots,v_{i-1}\}|=1$ and $N_G(v_{i-1})\cap D_{r-1}=\emptyset$, there exists a vertex $v_a\in N_G(v_{i-1})$ such that $v_a\in D^*$. Since the vertex $v_i$ is considered at $r$-th iteration, $v_a\notin D_{r-1}$. Now let $D'=(D^*\setminus\{v_a\})\cup\{v_i\}$. Since $D_{r-1}$ is a \ntd-set of $G[\{v_1,v_2,\ldots,v_{i-1}\}]$ by Corollary \ref{obspig4} and $i=n$, $D'$ is a minimum \ntd-set of $G$ containing $D_{r-1}\cup\{v_i\}$. If $i+1=n$, then, since $v_{i-1}v_{i+1}\notin E(G)$ and $v_i\notin D^*$, $v_{i+1}$ must be in $D^*$ to dominate $v_{i+1}$. If $v_{i+1}\in D_{r-1}$, then there must be a vertex $v_b$ for which $v_{i+1}$ is chosen by the algorithm. Since $v_i$ is considered by the algorithm at the $r$-th iteration, $b<i$. This implies that $v_{i-1}v_{i+1}\in E(G)$ since $v_{i+1}$ is a simplicial vertex in $G[\{v_1,v_2,\ldots,v_{i+1}\}]$. This is a contradiction and so $v_{i+1}\notin D_{r-1}$. Now $(D^*\setminus\{v_{i+1}\})\cup\{v_i\}$ is a minimum \ntd-set of $G$ containing $D_{r-1}\cup\{v_i\}$.

Now assume that $i\neq n$ and $i+1\neq n$. Then $\ell(v_{i+1})$ exists and in this case $D_r=D_{r-1}\cup\{v_p\}$, where $p=\ell(v_{i+1})$. If $v_p\in D^*$, then we are done. So assume that $v_p\notin D^*$. Since $v_i$ is dominated by $D_{r-1}$, by Lemma \ref{obspig2}, $v_{i-2}v_{i},v_{i-1}v_{i+1}\notin E(G)$. Let $v_c\in D^*$ be the minimum indexed vertex that dominates $v_{i+1}$. Since  $v_{i-1}v_{i+1}\notin E(G)$, $i\leq c\leq \ell(v_{i+1})$. Moreover, $v_c\notin D_{r-1}$. Now let $D''=(D^*\setminus\{v_c\})\cup \{v_p\}$. By Observation \ref{obsprelim5}(b), $N_{G_i}[v_i]\subseteq N_{G_i}[v_{i+1}]$ and $N_{G_{i+1}}[v_{i+1}]\subseteq N_{G_{i+1}}[v_{p}]$. This implies that $N_{G_{i+1}}[v_{c}]\subseteq N_{G_{i+1}}[v_{p}]$. Again since $v_i$ is dominated by $D_{r-1}$, $D''$ is a dominating set of $G$. Moreover, no vertex from $\{v_i,v_{i+1},\ldots,v_{p-1}\}$ is an isolated vertex in $G[N_G(D'')]$. So if there is an isolated vertex in $G[N_G(D'')]$, then it is a neighbor of $v_p$ that appears
after $p$. So as we have proved in the proof of Lemma \ref{pig2}, we can show considering $|N_{G_p} (v_p)|$ that either $D''$ is a minimum \ntd-set of $G$ or there exists a minimum \ntd-set of $G$ containing $D_{r-1}\cup \{v_p\}$. This is true since no vertex of $D_{r-1}$ will be removed
while modifying $D''$.

In each of the cases, we conclude that there is a minimum \ntd-set of $G$ containing $D_r$. This completes the proof of the lemma.\end{proof}

The proof of the correctness follows from Lemma \ref{pigcorrect1} and Lemma \ref{pigcorrect2}. We now discuss the running time of the algorithm. Suppose $G=(V,E)$ is a connected proper interval with $|V|=n$ and $|E|=m$. A BCO $\sigma=(v_1,v_2,\ldots,v_n)$ of a connected proper interval can be computed in $O(n+m)$ time \cite{panda-bco}. To implement the algorithm efficiently, we maintain an array \textsc{FDeg} and preprocess it. Given a BCO $\sigma=(v_1,v_2,\ldots,v_n)$ of $G$, for every vertex $v_i$ of $G$, \textsc{FDeg}[$i$] denotes the value $|N_G(v_i)\cap \{v_i,v_{i+1},\ldots,v_n\}|$. We compute  \textsc{FDeg}[$i$] for each $1\leq i\leq n$ as follows: Initially \textsc{FDeg}[$i$]=$d_G(v_i)$ for each $1\leq i\leq n$. We traverse the BCO $\sigma$ of $G$ and while traversing $\sigma$, for each $i,1\leq i\leq n$, \textsc{FDeg}[$j$] is decremented by $1$ for every $v_j\in N_{G_i}(v_i)$. So the array \textsc{FDeg} can be maintained in $O(n+m)$ time. We also maintain an array \textsc{Dom} to track whether a vertex is dominated or not by the so far constructed set $D$. Initially \textsc{Dom}[$i$]=$0$ for each $1\leq i\le n$. Once a vertex $v$ is selected into $D$ by the algorithm, then \textsc{Dom}[$j$] is made $1$ for every vertex $v_j\in N_G[v]$. So throughout the algorithm, the update of \textsc{Dom} can be done in $O(n+m)$ time. To store the value $\ell(v_i)$ for each $i, 1\leq i\leq n$, we maintain an array $L$, where $L[i]=j$ if $j=\max\{\{i\}\cup \{k:v_iv_k\in E\}\}$. This can be done in $O(n+m)$ time. So to find $\ell(v_{\ell(v_i)})$ for the vertex $v_i$ will take $O(1)$ time by looking into the array $L$.  Notice that in the algorithm, conditions like ``$|N_G(v_j)\cap \{v_j,v_{j+1},\ldots,v_n\}|\geq 2$" and ``$v_{j-1}v_{j+1}\in E$" are checked to decide the next iteration of the algorithm. These can be checked in constant time by looking the value of the vertices in the array \textsc{FDeg}. So \textsc{MNTDS-PIG($G$)} can be executed in $O(n+m)$ time. Therefore, we have the following theorem.

\begin{theorem}
A minimum \ntd-set of a given proper interval graph can be computed correctly in linear time.
\end{theorem}

\section{Hardness results}\label{sec_hard}

In this section, we discuss the lower bound and the upper bound of approximation ratio for \textsc{Min-NTDS} in graphs. For most of the terminologies used in this section, we refer to \cite{ausileo}.

\subsection{Lower bound on the approximation ratio}

To obtain the lower bound on \textsc{Min-NTDS}, we need the following result on \textsc{Min-Dom-Set}.

\begin{theorem}[\cite{chlebik,feige}]\label{hardness-thm1}
\textsc{Min-Dom-Set} cannot be approximated
within $(1-\varepsilon)\ln |V |$ for any $\varepsilon>0$, unless \textsf{NP $\subseteq$ DTIME($n^{O(\log\log n)}$)}.
\end{theorem}

\begin{theorem}\label{hardness-coro1}
Let $G=(V,E)$ be a graph with $n$ vertices. Unless \textsf{NP $\subseteq$ DTIME($n^{O(\log\log n)}$)}, \textsc{Min-NTDS} cannot be approximated
within a factor of $(1-\varepsilon) \ln n$ for any $\varepsilon>0$.
\end{theorem}

\begin{figure}[h]
\begin{center}
\includegraphics[scale=.8]{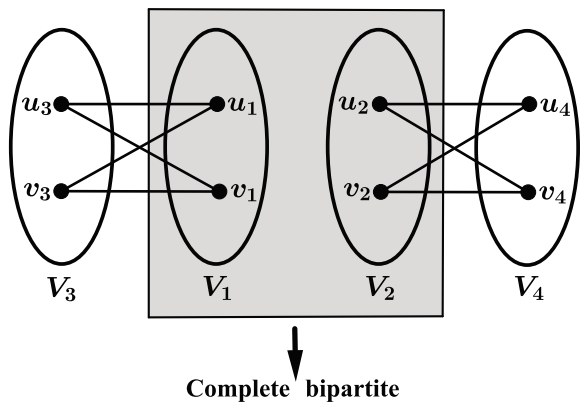}
\caption{\label{hard1} The construction of $G'$ used in Theorem \ref{hardness-coro1}}
\end{center}
\end{figure}

\begin{proof}
We prove this theorem by providing an approximation preserving reduction from \textsc{Min-Dom-Set} to \textsc{Min-NTDS}. Given a graph $G=(V,E)$, an instance of \textsc{Min-Dom-Set}, we construct a graph $G'=(V',E')$ as follows:
\begin{itemize}
\item[(1)] We take four copies $V_1,V_2,V_3$ and $V_4$ of the vertex set $V$ of $G$.
\item[(2)] For every edge $uv\in E$, we introduce the edges of the form $u_1v_3,u_3v_1,u_1u_3,v_1v_3,
u_2v_4,u_4v_2,u_2u_4,v_2v_4$. In other words, each vertex $x\in V_1$ (resp. $x\in V_2$) is made adjacent to a vertex $y\in V_3$ (resp. $y\in V_4$) if $y\in N_G[x]$.

\item[(3)] We introduce edges to make $G'[V_1\cup V_2]$ complete bipartite.
\end{itemize}
Notice that $G'$ is a bipartite graph with partite sets $V_1\cup V_4$ and $V_2\cup V_3$.

If $S$ is a dominating set of  $G$, then $ND=\{u_1,u_2:u\in S\}$ is a \ntd-set of $G'$ with $|ND|=2|S|$.  So in particular, we have $\gamma_{\rm nt}(G')\leq 2\gamma(G)$.

\begin{claim}\label{hard-claim1}
If $D'$ is a \ntd-set of $G'$, then either $D'\subseteq V_1\cup V_2$ or there is a \ntd-set $D_1$ of $G'$ with $|D_1|=|D'|$ such that $D_1\subseteq V_1\cup V_2$.
\end{claim}
\begin{proof}[Proof of Claim \ref{hard-claim1}:]
Let $D'$ be a \ntd-set of $G'$. If $D'\subseteq V_1\cup V_2$, then we are done. So assume that $D'\cap (V_3\cup V_4)\neq \emptyset$. Let $u_3\in V_3\cap D'$, where $u_3$ is the copy of the vertex $u$ in $V_3$. We consider the following cases.

\noindent\textbf{Case 1:} There exists a vertex $z\in N_{G'}(u_3)\cap V_1$ such that $z\notin D'$.

Let $D_2=(D'\setminus\{u_3\})\cup\{z\}$. If $N_{G'}(z)\cap D_2\neq\emptyset$, then $D_2$ is also a \ntd-set of $G'$. If $N_{G'}(z)\cap D_2=\emptyset$, then $D_2\cap V_2=D'\cap V_2=\emptyset$ since $G'[V_1\cup V_2]$ is complete bipartite. Now let $u_4\in V_4$ be the copy of $u$ in $V_4$. Since $V_4$ is independent and  $D_2\cap V_2=D'\cap V_2=\emptyset$, $u_4$ must be in $D'$ (and hence in $D_2$) to dominate itself. Now $D_3=(D_2\setminus\{u_4\})\cup \{u_2\}$ is also a \ntd-set of $G'$.

\noindent\textbf{Case 2:} $(N_{G'}(u_3)\cap V_1)\subseteq D'$.

Let $z'\in N_{G'}(u_3)\cap V_1\cap D'$. If $N_{G'}(z')\cap D'\neq\emptyset$, then $D'\setminus\{u_3\}$ is also a \ntd-set of $G'$ with cardinality less that $|D'|$. If $N_{G'}(z')\cap D'=\emptyset$, then $D'\cap V_2=\emptyset$ since $G'[V_1\cup V_2]$ is complete bipartite. Now let $u_4\in V_4$ be the copy of $u$ in $V_4$. Since $V_4$ is independent and  $D'\cap V_2=\emptyset$, $u_4$ must be in $D'$ to dominate itself. Now $D_4=(D'\setminus\{u_4\})\cup \{u_2\}$ is also a \ntd-set of $G'$.

By Case 1 and Case 2, we conclude that there is a \ntd-set $D''$ of $G'$ such that $D''\cap V_3=\emptyset$. Similarly arguing for every vertex $y\in V_4\cap D'$, we can obtain a \ntd-set $D_1$ of $G'$ such that $D_1\subseteq V_1\cup V_2$. This completes the proof of the claim.
\end{proof}

By Claim \ref{hard-claim1}, we can find a \ntd-set $D'$ of $G'$ such that $D'\subseteq V_1\cup V_2$. Let $S_k=V_k\cap D'$ for $k=1,2$. Without loss of generality, assume that $|S_1|\leq |S_2|$. Let $S'=\{v:v_1\in S_1\}$. Then $|S'|\leq \frac{1}{2}|D'|$. Since $V_3\cup V_4$ is an independent set of $G'$ and $D'$ is a \ntd-set of $G'$, $S'$ is a dominating set of $G$. So in particular, we have $\gamma (G)\leq |S'|\leq \frac{1}{2}\gamma_{\rm nt}(G')$.

Assume that \textsc{Min-NTDS} can be approximated with a ratio $\alpha$ by using a polynomial time algorithm $A$.  Consider the following algorithm:

\begin{algorithm}[H]
{\relsize{-1}
\KwIn{A graph $G=(V,E)$, an instance of \textsc{Min-Dom-Set}\;}
\KwOut{A dominating set $S$ of $G=(V,E)$\;}

Construct a graph $G'=(V',E')$ as described in (1), (2), and (3)\;
Compute a \ntd-set $D$ of $G'$ using algorithm $A$\;
Compute $D_k=D\cap V_k$ for $k=1,2$. Let $|D_1|\leq |D_2|$\;
Let $S=\{v:v_1\in D_1\}$\;
\Return $S$\;
\caption{$\boldsymbol{A'}$}}
\end{algorithm}

Let $D^*$ and $S^*$ be the
minimum \ntd-set of $G'$ and the minimum dominating set of $G$, respectively. It is clear that $|D^*|=2|S^*|$ as we have seen in the proof.  Algorithm $A'$ can compute a dominating set  of cardinality $|S|\leq \frac{1}{2}|D| \leq \frac{1}{2}\alpha\cdot |D^*|=\frac{1}{2}\alpha\cdot 2|S^*|=\alpha\cdot |S^*|$. Hence, algorithm $A'$ approximates \textsc{Min-Dom-Set} within ratio $\alpha$. So by Theorem \ref{hardness-thm1}, unless \textsf{NP $\subseteq$ DTIME($n^{O(\log\log n)}$)}, \textsc{Min-NTDS} cannot be approximable within a factor of $(1-\varepsilon)\ln n$ for any $\varepsilon>0$. This completes the proof of the theorem.
\end{proof}

\subsection{Upper bound on the approximation ratio}

In this section, we present a $2(\ln (\Delta(G)+1)+1)$-factor approximation algorithm
for \textsc{Min-NTDS} in a given graph $G$. The main idea behind the algorithm is that given a graph $G=(V,E)$,
we first compute a dominating set $D$ of $G$ and then we try to include a set $S$ of vertices from $V\setminus D$ such that $|S|$ is as small as possible and  $D\cup S$ forms a \ntd-set
of $G$. 

\begin{algorithm}[H]
{\relsize{-1}
\KwIn{A connected graph $G=(V,E)$\;}
\KwOut{A \ntd-set of $G$\;}
$S=\emptyset$\;
Compute a dominating set $D$ of $G$ applying the procedure described in \citep{chavtal}\;
Let $D_1=\{v\in D: N_G(v)\cap D=\emptyset\}$, $D_2=\{v\in D: N_G(v)\cap D\neq \emptyset\}$ and $V'=V\setminus (D_1\cup N_G(D_2))$\;
\While{$(D_1\neq \emptyset)$}{Choose a vertex $u\in D_1$ and then choose a vertex $u'\in N_G(u)\cap V'$ (if exists) such that $N_G(u')\cap (V\setminus D)=\emptyset$\;
$S=S\cup \{u'\}$\;
$V'=V'\setminus N_G(A)$, where $A=D_1\cap N_G(u')$\;
$D_1=D_1\setminus N_G(u')$\;
}
\Return $D\cup S$\;

\caption{\label{approx-algo} \textsc{Approx-NTDS$(G)$}}}
\end{algorithm}

\begin{lemma}\label{approx1}
The constructed set $D\cup S$ is a \ntd-set of $G$.
\end{lemma}
\begin{proof}
To prove that $D\cup S$ is a \ntd-set of $G$, we show that $G[N_G(D\cup S)]$ has no isolated vertex. Notice that $D\cup S=D_1\cup D_2\cup S$ and $G[N_G(D_2)]$ has no isolated vertex. Let $V'=V\setminus(D_1\cup N_G(D_2))$. For $x\in D_1$, let $x'\in N_G(x)\cap V'$ be a vertex such that $N_G(x')\cap (V\setminus D)=\emptyset$. Notice that $x'$ may exist or may not exist. If for some $x\in D_1$, there does not exist $x'$, then $x\notin G[N_G(D\cup S)]$. So by choosing the set $S$ constructed by the algorithm, we can see that $G[N_G(D_1\cup S)]$ has no isolated vertex. This implies that $D\cup S$ is a \ntd-set of $G$.
\end{proof}

Notice that all the steps of \textsc{Approx-NTDS($G$)} can be executed in polynomial time. Let $ND^*$ be a minimum \ntd-set of $G$. If $D^*$ is a minimum  dominating set of $G$,
then $|ND^*|\geq |D^*|$ by Observation \ref{obs1}. Notice that a dominating set $D$ of $G$ can be computed in polynomial time within an approximation ratio of $\ln (\Delta(G)+1)+1$ \cite{chavtal}. If $S=\emptyset$, then $D$ is also a \ntd-set of $G$ and hence $\frac{|D|}{|ND^*|}\leq \frac{|D|}{|D^*|}\leq \ln (\Delta(G)+1)+1<2 (\ln (\Delta(G)+1)+1)$. If $S\neq \emptyset$, then by Lemma \ref{approx1}, $D\cup S$ is a \ntd-set of $G$. Hence

$$\frac{|D\cup S|}{|ND^*|}=\frac{|D|+|S|}{|ND^*|}=\frac{|D|}{|ND^*|}+\frac{|S|}{|ND^*|}\leq
2\cdot\frac{|D|}{|ND^*|}\leq 2\cdot\frac{|D|}{|D^*|}.
$$

Since finding a minimum dominating set in $G$ can be approximated within a factor of $\ln (\Delta(G)+1)+1$ \cite{chavtal}, we have $\dfrac{|D\cup S|}{|ND^*|}\leq 2 (\ln (\Delta(G)+1)+1)$ and hence we have the following theorem.

\begin{theorem}\label{theo-approx}
\textsc{Min-NTDS} in a graph $G=(V,E)$ can be approximated within an approximation ratio of $2(\ln (\Delta(G)+1)+1)$.
\end{theorem}

\section{\textsf{APX}-completeness}\label{sec_apx}

In this section, we show that \textsc{Min-NTDS} is \textsf{APX}-complete for graphs of maximum degree $3$.

To show the \textsf{APX}-completeness of a problem $\pi\in$
\textsf{APX}, it is enough to show that there is an $L$-reduction
from some \textsf{APX}-complete problem to $\pi$.

We first recall the notion of $L$-reduction \cite{ausileo,papadi}.
Given two \textsf{NP}-optimization problems $\pi_1$ and $\pi_2$, and a
polynomial time transformation $f$ from instances of $\pi_1$ to
instances of $\pi_2$, we say that $f$ is an \emph{$L$-reduction} if there
are positive constants $a$ and $b$ such that for every
instance $x$ of $\pi_1$:
\begin{enumerate}
\item $opt_{\pi_2}(f(x))\leq a\cdot opt_{\pi_1}(x)$;
\item for every feasible solution $y$ of $f(x)$ with objective value $m_{\pi_2}(f(x),y)=c_2$,
we can find a solution $y'$ of $x$ in polynomial time with
$m_{\pi_1}(x,y')=c_1$ such that $|opt_{\pi_1}(x)-c_1|\leq b\cdot
|opt_{\pi_2}(f(x))-c_2|$.
\end{enumerate}

We use the fact that \textsc{Min-Dom-Set} is \textsf{APX}-complete for graphs with degree at most $3$ \cite{chlebik} to show that  \textsc{Min-NTDS} is \textsf{APX}-complete for graphs with degree at most $3$. We show this by providing an $L$-reduction from \textsc{Min-Dom-Set} for graphs with degree at most $3$ to \textsc{Min-NTDS} for graphs with degree at most $3$.

\begin{figure}[h]
\begin{center}
\includegraphics[scale=.8]{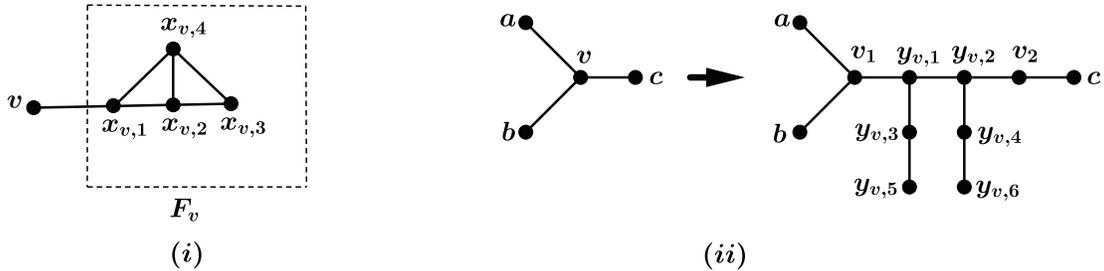}
\caption{\label{apx-fig} $(i)$ The structure $F_v$ in $G'$ for a vertex $v$ of $G$ with $d_G(v)\leq 2$, $(ii)$ Transformation of a vertex $v$ of $G$ of degree $3$ in $G'$ }
\end{center}
\end{figure}

\begin{theorem}
\textsc{Min-NTDS} is \textsf{APX}-complete for graphs of degree at most $3$.
\end{theorem}
\begin{proof}
By Theorem \ref{theo-approx}, \textsc{Min-NTDS} is in \textsf{APX}. We show that \textsc{Min-NTDS} for graphs of degree at most $3$ is \textsf{APX}-complete by providing an $L$-reduction $f$ from \textsc{Min-Dom-Set} for graphs of degree at most $3$ to it. Let $G=(V,E)$ be a graph having $n$ vertices that is an instance of \textsc{Min-Dom-Set} for graphs of degree at most $3$. We construct a graph $G'=(V',E')$, an instance of \textsc{Min-NTDS} as follows. Let $v$ be an arbitrary vertex of $G$. If $d_G(v)\leq 2$, then attach a $F_v$ as shown in Figure \ref{apx-fig}$(i)$ at $v$ by adding the edge $vx_{v,1}$. If $d_G(v)=3$, then $v$ is split and transformed as shown in Figure \ref{apx-fig}$(ii)$. Here we say $v_1,v_2$ as the copies of $v$ in $G'$. Notice that $\Delta(G')=3$.

Let $D$ be a dominating set of $G$. Now we construct a set $ND$ that is a subset of $V'$ as follows.
\begin{itemize}
\item[1.] If $d_G(v)\leq 2$ and $v\in D$, then include $v$ and $x_{v,2}$ in $ND$.
\item[2.] If $d_G(v)\leq 2$ and $v\notin D$, then include $x_{v,2}$ in $ND$.
\item[3.] If $d_G(v)=3$ and $v\in D$, then include $v_1,v_2,y_{v,5},y_{v,6}$ in $ND$.
\item[4.] If $d_G(v)=3$ and $v\notin D$, then at least one of $a,b,c$ must be in $D$. Now do the following.
\begin{itemize}
\item[4.1.] If at least one of $a$ and $b$ is in $D$ and $c\notin D$, then include $y_{v,2},y_{v,4},y_{v,5}$ in $ND$.
\item[4.2.] If exactly one of $a$ and $b$ is in $D$ and $c\in D$, then include $y_{v,2},y_{v,4},y_{v,5}$ in $ND$.
\item[4.3.] If $a,b\notin D$ and $c\in D$, then include $y_{v,1},y_{v,3},y_{v,6}$ in $ND$.
\item[4.4.] If $a,b,c\in D$, then include $y_{v,2},y_{v,4},y_{v,5}$ (or $y_{v,1},y_{v,3},y_{v,6}$) in $ND$.
\end{itemize}
\end{itemize}
\begin{claim}\label{apxclaim1}
The set $ND$ is a \ntd-set of $G'$.
\end{claim}
\begin{proof}[Proof of Claim \ref{apxclaim1}:]
It can be seen that $ND$ is a dominating set of $G'$.
By Observation \ref{obsprelim3}, each vertex of $ND$ is not an isolated vertex in $G'[N_{G'}(ND)]$. So for each vertex $v$ of $V'\setminus ND$, we show that either $N_{G'}(v)\cap (V'\setminus ND)\neq \emptyset$ or $N_{G'}(u)\cap ND\neq \emptyset$, where $u\in N_{G'}(v)\cap ND$. Notice that $N_{G'}(v)\cap ND\neq\emptyset$ since $ND$ is a dominating set of $G'$.

Let $v\in V'\setminus ND$ be arbitrary. If $v\notin V$ or $v$ is not the copy of any vertex of $G$ in $G'$, then it is clear that either $N_{G'}(v)\cap (V'\setminus ND)\neq \emptyset$ or there exists a vertex $u\in N_{G'}(v)\cap ND$ such that $N_{G'}(u)\cap ND\neq \emptyset$. So assume that $v\in V\setminus ND$.

First assume that $d_G(v)\leq 2$. Since $v\notin ND$ and $x_{v,2}\in ND$, $x_{v,1}\in N_{G'}(v)\cap (V'\setminus ND)$. Next assume that $d_G(v)=3$.   Let $N_G(v)=\{a,b,c\}$ and $v_1,v_2$ are the copies of $v$ in $G'$. If $a\in D$ or $b\in D$, then $y_{v,1}\in N_{G'}(v_1)\cap (V'\setminus ND)$, $y_{v,2}\in N_{G'}(v_2)\cap ND$, and $y_{v,4}\in N_{G'}(y_{v,2})\cap ND$. If $a,b\notin D$, then $c\in D$. In this case, $y_{v,1}\in N_{G'}(v_1)\cap ND$, $y_{v,3}\in N_{G'}(y_{v,1})\cap ND$, and $y_{v,2}\in N_{G'}(v_2)\cap (V'\setminus ND)$. If exactly one of $a$ and $b$, say $a$ is in $D$ and $c\in D$, then $y_{v,1}\in N_{G'}(v_1)\cap (V'\setminus ND)$, $y_{v,2}\in N_{G'}(v_2)\cap ND$, and $y_{v,4}\in N_{G'}(y_{v,2})\cap ND$. This completes the proof of the claim.
\end{proof}

Let $ND'$ be a \ntd-set of $G'$. Let $T(v)=\{v,x_{v,1},x_{v,2},x_{v,3},x_{v,4}\}$ if $d_G(v)\leq 2$ and $T(v)=\{v_1,v_2,y_{v,1},y_{v,2},y_{v,3},y_{v,4},y_{v,5},y_{v,6}\}$ if $d_G(v)=3$. Let $\rho(v)=|T(v)\cap ND'|$ for every vertex $v$ of $G$. The following facts can be verified easily.

\begin{fact}\label{apxfact1}
If $d_G(v)\leq 2$ and $v\in ND'$, then $\rho(v)\geq 2$.
\end{fact}

\begin{fact}\label{apxfact2}
If $d_G(v)=3$ and $v_1$ or $v_2$ is in $ND'$, then $\rho(v)\geq 4$.
\end{fact}

Let $s$ be the number of vertices of $G$ having degree $3$.
\begin{claim}\label{apxclaim2}
A dominating set $D'$ of $G$ can be constructed from $ND'$ such that $|D'|\leq |ND'|-(n-s)-3s$.
\end{claim}
\begin{proof}[Proof of Claim \ref{apxclaim2}:]
We construct a set $D'$ from $ND'$ as follows. If $d_G(v)\leq 2$, then we include $v$ in $D'$ if and only if $\rho(v)\geq 2$. If $d_G(v)=3$, then we include $v$ in $D'$ if and only if $\rho(v)\geq 4$. We now prove that $D'$ is a dominating set of $G$. Suppose that there is a vertex $z\in V(G)$ such that $N_{G}[z]\cap D'=\emptyset$. This implies that if $d_G(z)\leq 2$, then $\rho(z)\leq 1$. Thus, it can be obtained that $z$ is not dominated by any vertex of $T(z)$. Since $ND'$ is a dominating set of $G'$, there must be a vertex other than $x_{v,1}$, say $w\in ND'$ that dominates $z$. Let $w\in T(w)$. By Fact \ref{apxfact1} and Fact \ref{apxfact2}, $\rho(w)\geq 2$ if $d_G(w)\leq 2$ and $\rho(w)\geq 4$ if $d_G(w)=3$. This implies that $w$ must have been included in $D'$. This contradicts our assumption that $N_G[z]\cap D'=\emptyset$. Similarly we get a contradiction if $d_G(z)=3$. So $D'$ is a dominating set of $G$. On the other hand, it can be seen that for every $v\in V(G)$, $\rho(v)\geq 1$ if $d_G(v)\leq 2$ and $\rho(v)\geq 3$ if $d_G(v)=3$. Thus, $|D'|\leq |ND'|-(n-s)-3s$.
\end{proof}

By Claim \ref{apxclaim1}, $ND$ is a \ntd-set of $G'$ with cardinality $|ND|=|D|+(n-s)+3s$, where $s$ is the number of vertices of $G$ having degree $3$. In particular, $\gamma_{\rm nt}(G')\leq \gamma(G)+(n-s)+3s\leq \gamma(G)+3n$. Since $\Delta(G)=3$, we have $\gamma(G)\geq \frac{n}{4}$ \cite{haynesb1}. So $\gamma_{\rm nt}(G')\leq \gamma(G)+3n\leq \gamma(G)+12\cdot\gamma(G)=13\cdot\gamma(G)$.

On the other hand, given a \ntd-set $ND'$ of $G'$, by Claim \ref{apxclaim2}, a dominating set $D'$ of $G$ can be constructed such that $|D'|\leq |ND'|-(n-s)-3s=|ND'|-n-2s$, where $s$ is the number of vertices of $G$ having degree $3$. In particular, $\gamma(G)\leq \gamma_{\rm nt}(G')-n-2s$. Therefore, $\gamma_{\rm nt}(G')=\gamma(G)+n+2s$. Now $|D'|-\gamma(G)\leq |ND'|-n-2s-(\gamma_{\rm nt}(G')-n-2s)=|ND'|-\gamma_{\rm nt}(G')$. So $f$ is an $L$-reduction with $a=13$ and $b=1$ and hence \textsc{Min-NTDS} is \textsf{APX}-complete for graphs of degree at most $3$.
\end{proof}

\end{document}